\documentclass[11pt]{article}

\usepackage[numbers]{natbib}

\newcommand{\nonstoc}[1]{#1}
\usepackage{fullpage}
\usepackage{float}
\usepackage{epsfig}
\usepackage{amsmath}
\usepackage{amssymb}
\usepackage{amstext}
\usepackage{amsmath}
\usepackage{verbatim}
\usepackage{url}
\usepackage[ruled,vlined]{algorithm2e}

\newtheorem{theorem}{Theorem}[section]
\newtheorem{definition}{Definition}
\newtheorem{claim}{Claim}
\newtheorem{lemma}[theorem]{Lemma}

\newtheorem{corollary}[theorem]{Corollary}

\newtheorem{observation}{Observation}

\newcommand{\qed}{\mbox{\ \ \ }\rule{6pt}{7pt} \bigskip}

\renewcommand{\comment}[1]{}
\newenvironment{proof}{\noindent{\em Proof:}}{\hfill\qed\\}

\newenvironment{proofsketchof}[1]{\noindent{\em Proof Sketch of #1:}}{\hfill\qed\\}
\newenvironment{proofof}[1]{\noindent{\em Proof of #1:}}{\hfill\qed\\}

\newcommand{\argmax}{\operatorname{argmax}}

\newcommand{\val}[1]{v_{#1}}

\newcommand{\Aa}{A1}
\newcommand{\Ab}{A2}
\newcommand{\Ac}{A3}
\newcommand{\Ad}{A4}

\newcommand{\utiltarget}{\pi}

\newcommand{\outcome}{o}
\newcommand{\optoutcome}{o^*}
\newcommand{\sbid}[1]{\utiltarget_{#1}}
\newcommand{\optsbid}[1]{\utiltarget^*_{#1}}

\newcommand{\level}{L}
\newcommand{\cefset}{\mathcal{C}}

\newcommand{\ncefset}{\overline{\cefset}}

\newcommand{\allsbid}{\mathrm{\utiltarget}}
\newcommand{\nearcef}{\cefset_\epsilon}

\newcommand{\nearncef}{\ncefset_\epsilon}
\newcommand{\numbidders}{n}
\newcommand{\bidders}{[\numbidders]}
\newcommand{\outcomes}{\mathcal{O}}
\newcommand{\altbidderset}{\mathcal{B}}

\newcommand{\witness}{\outcome^{w}}
\newcommand{\bidderidx}{j}
\newcommand{\altbidderidx}{j'}
\newcommand{\levelidx}{i}
\newcommand{\lowlevelidx}{\levelidx^{-}}
\newcommand{\highlevelidx}{\levelidx^{+}}
\newcommand{\altlevelidx}{\levelidx'}
\newcommand{\lowbound}{b^-}
\newcommand{\highbound}{b^+}
\newcommand{\totalbid}{\mathbf{B}}

\begin{document}
\title{A Dynamic Axiomatic Approach to First-Price Auctions}

\author{Darrell Hoy\thanks{Northwestern University, Evanston, IL. Much of this work was done while the author was an intern at eBay Research Labs. }
\and Kamal Jain\thanks{eBay Research Labs, San Jose, CA.}
\and Christopher A. Wilkens\thanks{University of California at Berkeley, Berkeley, CA. Much of this work was done while the author was an intern at eBay Research Labs. This work was supported in part by NSF award CCF-0964033}
}

\date{}

\maketitle{}

\begin{abstract} The first-price auction is popular in practice for its simplicity and transparency. Moreover, its potential virtues grow in complex settings where incentive compatible auctions may generate little or no revenue. Unfortunately, the first-price auction is poorly understood in theory because equilibrium is not {\em a priori} a credible predictor of bidder behavior.

We take a dynamic approach to studying first-price auctions: rather than basing performance guarantees solely on static equilibria, we study the repeated setting and show that robust performance guarantees may be derived from simple axioms of bidder behavior. For example, as long as a loser raises her bid quickly, a standard first-price auction will generate at least as much revenue as a second-price auction.

We generalize this dynamic technique to complex pay-your-bid auction settings: as long as losers do not wait too long to raise bids, a first-price auction will reach an envy-free state that implies a strong lower-bound on revenue; as long as winners occasionally experiment by lowering their bids, the outcome will near the boundary of this envy-free set so bidders do not overpay; and when players with the largest payoffs are the least patient, bids converge to the egalitarian equilibrium. Significantly, bidders need only know whether they are winning or losing in order to implement such behavior.

Along the way, we find that the auctioneer's choice of bidding language is critical when generalizing beyond the single-item setting, and we propose a specific construction called the {\em utility-target auction} that performs well. The utility-target auction includes a bidder's final utility as an additional parameter, identifying the single dimension along which she wishes to compete. This auction is closely related to profit-target bidding in first-price and ascending proxy package auctions and gives strong revenue guarantees for a variety of complex auction environments. Of particular interest, the guaranteed existence of a pure-strategy equilibrium in the utility-target auction shows how Overture might have eliminated the cyclic behavior in their generalized first-price sponsored search auction if bidders could have placed more sophisticated bids.
\end{abstract}



\section{Introduction}

In 1961, Vickrey~\cite{V61} initiated the formal study of auctions. He first considered common auctions of the day --- including  the first-price auction, the Dutch auction, and the English auction --- and studied their equilibria. Vickrey observed that the English auction was, in theory, more robust because each player had a strategy that dominated all others regardless of other players' bids. As a solution, he proposed\footnote{While Vickrey was the first to discover the second-price auction in the economics literature, it has been used in practice as early as 1893~\cite{L00}.} the second-price auction as a means to achieve the same robustness in a sealed-bid format. The subsequent development of auction theory largely followed Vickrey's paradigm: existing auctions were evaluated in terms of their equilibria, meanwhile the field of mechanism design emerged with dominant strategy incentive compatibility as a {\em sine qua non}.

Fifty years later, it is apparent that Vickrey's analysis does not always give best guide to implementing a real auction. In mechanisms without dominant strategies, Vickrey's original concern still stands --- equilibrium is a highly questionable predictor of outcome due (at least in part) to players' informational limitations~\cite{V61,H89}. Neither is dominant strategy incentive compatibility a magic solution: incentive compatible mechanisms have sufficiently many drawbacks that their real attractiveness rarely matches theory --- the simple and elegant second-price auction has earned the title ``Lovely but Lonely''~\cite{AM06} for its sparse use. Even the supposition that bidders will play strategies that are theoretically dominant is discredited by a wide variety of practical issues~\cite{K02}.

Dynamic analysis offers a powerful complement to Vickrey's static approach. For example, certain behavior will be clearly irrational when an auction is repeated. Such reasoning was used by Edelman and Schwarz~\cite{ES10} in the context of the generalized second-price (GSP) ad auction --- they analyzed a dynamic game to derive bounds on reasonable outcomes of the auction, then studied the static game under the assumption that these bounds were satisfied. Dynamic settings also introduce new pitfalls: Edelman and Ostrovsky~\cite{EO07} showed that the instability of Overture's generalized first-price (GFP) ad auction could be attributed to its lack of a pure-strategy equilibrium.

We study repeated first-price auctions and show that they offer powerful performance guarantees. We begin with a static perspective and observe that the equilibrium properties of the auction depend significantly on the types of bids that bidders can express. We propose a generalization of the first-price auction called the {\em utility-target auction} that is closely related to profit-target bidding in first-price and ascending proxy package auctions~\cite{M04, DM07}. Like these package auctions, we show that the utility-target auction possess many advantages over incentive compatible mechanisms in a static equilibrium analysis, including revenue, simplicity, and transparency. More significantly, we show that {\em the same performance guarantees may be derived using only a few simple behavioral axioms and limited information} in a repeated setting. These dynamic results are particularly powerful because they do not require an a priori assumption that the auction will reach equilibrium --- for example, assuming only that losers will not wait too long to raise their bids {\em the auctioneer's revenue satisfies a natural lower bound regardless of whether bidders' behavior converges to equilibrium}. Moreover, bidders need only know if they are winning or losing to implement the dynamics. We build on these axioms to demonstrate behavior that offers progressively stronger performance guarantees, culminating with a set of axioms that together imply convergence to the egalitarian equilibrium.

\paragraph{First-Price Auctions\nonstoc{.}} The virtues of a first-price auction --- and other auctions in the pay-your-bid family --- arise from its simplicity. From the bidders' perspective, the pay-your-bid property offers transparency, credibility, and privacy: not only is the auction easy to understand, but it ensures that the auctioneer cannot cheat (say, by unreasonably inflating the reserve price in a repeated auction) and allows a bidder to participate without expressing her true willingness to pay.

The auctioneer can also benefit from this simplicity because players' bids represent guaranteed revenue. By comparison, the revenue from a dominant strategy incentive compatible auction is almost always less than the bids and, in the most general settings, may even be zero~\cite{AM06,Rob79}. Supposing a first-price auction reaches equilibrium, a variety of work presents settings where they generate more revenue for the seller than their incentive compatible brethren~\cite{M04,LST12,HJW} (though they may also generate less revenue~\cite{MR00}).

Yet, running a first-price auction is risky. While first-price auctions have been quite successful in settings like treasury bill and procurement auctions, Overture's generalized first-price (GFP) auction for sponsored search advertising was erratic: bids rapidly rose and fell in a sawtooth pattern, rendering the auction unpredictable and depressing revenue~\cite{EO07}. As a result, the sponsored search industry has moved to a generalized second-price (GSP) auction that leverages the intuition of the second-price auction to disincentivize small adjustments to a player's bid. 

The challenges of a first-price auction are many and complex. Vickrey identified a major source of risk in the first-price single-item auction: since a rational bidder's optimal bid depends on other players' bids, actual behavior will depend on beliefs about others' strategies. A first-price auction also requires bidders to strategize, a task that is may be difficult and expensive. At best, players will be in a Bayesian equilibrium, and, at worst, they will be completely unpredictable. Indeed, predicting the outcome of a first-price auction lies at the center of a lively debate between experimental and theoretical economists~\cite{H89}.

Experience with GFP highlights another potential pitfall of first-price auctions: when generalized beyond the single-item setting, a first-price auction may not have a pure-strategy equilibrium. Edelman and Ostrovsky~\cite{EO07} showed that this was the case with GFP and demonstrated how it generated the rapid sawtooth behavior seen in practice. Our goal is to demonstrate that how proper design coupled with dynamic arguments can support strong performance guarantees.

\paragraph{The Utility-Target Auction\nonstoc{.}} The equilibria and performance of a pay-your-bid auction depend on its implementation. Within the pay-your-bid constraint, the auctioneer chooses the form of players' bids, potentially restricting or broadening the bids that players may express.

The historical performance of the GFP ad auction exemplifies the importance of choosing a good bidding language. In the GFP auction, advertisers placed a single bid and paid the bid price for each click regardless of where their ads were shown. In retrospect, the rapid sawtooth motion observed in bids is not surprising because the auction had no pure-strategy equilibrium~\cite{EOS07,EO07}; however, we show that a pure-strategy equilibrium would have existed if bidders could have placed more expressive bids, such as bidding different prices for each slot.

A natural question arises: what are good bidding languages and how complicated must a language be to offer good performance? In GFP, the bidding language is precisely sufficient to represent any possible valuation function; hence, it is possible that bids may need to be more expressive than the space of valuation functions.

We show that the overhead required for a good bidding language is at most a single value: it is sufficient to ask bidders for their valuation function and their final desired utility. We call such an auction a {\em utility-target auction}: a player's bid includes a specification of her value for every outcome and a single number representing the utility-target that she requests regardless of the outcome. Her payment is her claimed value for the final outcome minus the utility-target that she requested, and the auctioneer chooses the outcome that maximizes the total payment. In essence, the utility-target auction isolates the single dimension (utility) along which a bidder truly wishes to strategize.

We begin with a static analysis of the utility-target auction's equilibria. We first show that the utility-target auction is quasi-incentive compatible: a bidder never has an incentive to misreport her valuation function --- it is always sufficient for her to manipulate the utility-target she requests. Moreover, we show that a pure-strategy equilibrium always exists and that the egalitarian equilibrium is efficiently computable. These results are closely related to profit-target equilibria in package auctions~\cite{M04}.

Next, we show that the utility-target auction offers good equilibrium performance. Similar to the approach of Edelman, Ostrovsky, and Schwarz~\cite{EOS07} on the generalized second-price (GSP) auction, we show that all equilibria satisfying a natural envy-free criterion have good performance. First, such equilibria are efficient and generate at least as much revenue as the incentive compatible Vickrey-Clarke-Groves (VCG) mechanism. Moreover, they generate revenue even when the incentive compatible mechanisms fail --- the revenue of the envy-free equilibria of a utility-target auction all meet an intuitive benchmark we call the second-price threat, even settings where the VCG mechanism may make little or no revenue. Again, this bound is related to the core property of profit-target equilibria in package auctions~\cite{M04, DM07}.

\paragraph{Dynamic Analysis through Behavioral Axioms\nonstoc{.}} A significant novelty of our work is our use of simple behavioral axioms to prove guarantees on the performance of utility-target auctions.

Dynamic arguments are generally fraught with peril: in addition to being difficult to prove, more complex auctions (or markets, or games) require more complex bidding behavior to converge to an equilibrium and therefore sacrifice robustness. For example, Walrasian t\^atonnement\footnote{To justify market equilibrium as a predictor of actual market behavior, Leon Walras described a dynamic procedure called t\^atonnement that might converge to it.}~\cite{W54} is perhaps the earliest concrete dynamic procedure proposed in economics --- it converges in general markets when modeled as a particular continuous process~\cite{S41,ABH59} but may or may not converge as a discrete process~\cite{BM92,CF08}. More recent results have sought stronger guarantees, e.g. by showing that players' behavior will converge to equilibrium in repeated games as long as their learning strategies are ``adaptive and sophisticated''~\cite{MR91} or no-regret~\cite{HM00,EMN09}. However, these properties are sufficiently complicated that it is difficult to evaluate whether players' strategies indeed satisfy them in practice.

In contrast, we build simple behavioral axioms and use them to prove performance guarantees. Our first axioms are that (a) a bidder who is losing will raise her bid to try to win, and (b) a bidder who is losing is more impatient than a bidder who is winning. After formalizing these axioms in the context of utility-target auctions, we show that the auction will eventually reach an outcome that satisfies a natural notion of envy-freeness and, by extension, a natural second-price type bound on revenue. Significantly, this result neither implies nor requires that players' bids converge to a steady-state. Moreover, bidder behavior requires only knowing whether one is winning or losing, not the precise bids of other players.

Next, we show that bidders will not overpay if two more axioms are also satisfied, namely that (c) bidders who are winning will try to lower their bid to save money. Axioms (a)-(c) guarantee that bids will ultimately remain close to the boundary between envy-free and non-envy-free outcomes, a boundary which contains the envy-free equilibria. These axioms offer a degree of robustness, since bids will seek this boundary even as bids and ads change.

Finally, we show that bids will converge to the egalitarian equilibrium --- the equilibrium that distributes utility most evenly --- if a fourth axiom is satisfied. The fourth axiom concerns the timing of raised bids: (d) the bidder who has the most value at risk is the least patient and therefore raises her bid first. When bidder behavior satisfies all five axioms (a)-(d), we show that bids will converge to the egalitarian equilibrium. Together, these results offer powerful guarantees about the performance of a utility-target auction in a repeated setting.

\paragraph{Related Work\nonstoc{.}} Our utility-target auction is most closely related to first-price package auctions \cite{BW86} and the ascending proxy auction~\cite{M04}. Profit-target bidding in these auctions is closely related to quasi-truthful bidding in utility-target auctions, and the static properties we prove in Section~\ref{sec:eq} all have direct analogues. In contrast, the utility-target auction can be applied beyond the package auction setting (e.g. to ad auctions), and our dynamic analysis is entirely new, a more general confirmation of Milgrom's postulate that profit-target equilibria ``may describe a central tendency for some kinds of environments''~\cite{M04}.

Auctions in which a player's bid directly specifies her payment are known as {\em pay-your-bid} auctions. The first-price auction, as well as the Dutch an English auctions, are members of this family. Our utility-target auction is closely related to first-price and ascending proxy package auctions~\cite{M04}. Engelbrecht-Wiggans and Kahn~\cite{EK98} explored multi-unit, sealed-bid pay-your-bid auctions and found their equilibria to be substantially different from the standard first-price auction --- the core issue they encounter is the same one arising in GFP.

A key reason repeated auctions may admit more robust performance guarantees is that bidders can learn about others' valuations. A similar informational exchange is present in and a motivation for classic ascending auctions. In addition to his discussion of ascending proxy auctions~\cite{M04}, Milgrom offers a broad discussion of this literature in~\cite{M00}. Some recent work studies ascending auctions for position auctions like sponsored search~\cite{EOS07,ABHLT10}.

Our work can also be seen through the lens of {\em simple versus optimal} mechanisms~\cite{HR09}. The general goal of this line of research is to design a mechanism that is simple and transparent while (possibly) sacrificing efficiency or revenue. For example, Hart and Nisan analyze the tradeoff between the number of different bundles offered to a buyer and an auction's performance~\cite{HN12}. By comparison, our results show that a first-price auction can guarantee good performance when the bidding complexity is only slightly larger than that of the valuation functions.


\newcommand{\ex}{\mathbf{E}}

\section{Definitions and Preliminaries}\label{sec:s-prelim}

The {\em utility-target auction} is a generalization of the first-price auction. Its key feature is an extra utility-target parameter in the bid --- this parameter highlights the key dimension along which bidders care to compete. It gives bidders sufficient flexibility to guarantee the existence of pure-strategy equilibria while minimizing the communication required between the bidders and the auctioneer.

\subsection{First-Price and Pay-Your-Bid Auctions\nonstoc{.}} An auction is a protocol through which players bid to select an outcome. A standard sealed-bid auction can be decomposed into three stages: (1) each player $i$ submits a bid $b_i$, (2) the auctioneer uses players' bids to pick an outcome $o$ from a set $\outcomes$, and finally (3) each player $i$ pays a price $p_i$. The final utility of player $i$ is given by $v_i(o)-p_i$, where $v_i(o)\geq 0$ denotes $i$'s value for the outcome $o$, i.e. $v_i\in V_i$ is $i$'s valuation function (drawn from a publicly known set $V_i$).

From this perspective, the standard first-price auction is described as follows: (1) each player submits a single number $b_i\in\Re$, (2) the auctioneer chooses to give the item to the player $i^*$ who submits the largest bid $b_i$, and (3) the winner $i^*$ pays $b_{i^*}$ and everyone else pays zero. For comparison, the second-price auction is identical to the first-price auction except that the price paid is equal to the second-highest value of $b_i$.

When the outcomes are few, we will use $v_i$ and $b_i$ to denote the profile of values and bids across outcomes, e.g., $v_i=(1, 1.5, 0)$.

When considering settings beyond the single-item auction there are many ways to generalize the first-price auction. Even within the single-item setting, the auctioneer could choose an arbitrary encoding for players' bids. Moreover, the auctioneer might choose an encoding that changes the space of possible bids, e.g by forcing bidders to place integer bids when values are actually real numbers. In such cases, the principle feature that we wish to preserve is that the winner ``pays what she bid,'' or alternatively that a player's bid precisely specifies her payment. Formally, we say that such an auction has the pay-your-bid property:
\begin{definition}
An auction has the {\em pay-your-bid property} if the payment $p_i$ depends only on the outcome $o$ and $i$'s bid $b_i$ (it does not directly depend on others' bids).
\end{definition}
The first-price auction as described above clearly satisfies this property while a second-price auction does not.

Not all sealed-bid pay-your-bid auctions are equivalent. Edelman et al.~\cite{EOS07} showed that GFP, where the set of possible bids is precisely $V_i$, did not have a pure-strategy equilibrium:
\begin{observation}
The pay-your-bid property does not guarantee the existence of a pure-strategy equilibrium in a sealed-bid auction when the space of bids is the same as the space of valuation functions.
\end{observation}
Moreover, as we discuss in Section~\ref{sec:adauction}, any pay-your-bid ad auction where bids are restricted to a subset of $V_i$ must suffer in terms of its welfare and revenue guarantees. Thus, it is import to consider auctions that allows bids $b_i\not\in V_i$. This motivates us to introduce the utility-target auction, a sealed-bid pay-your-bid auction that allows such bids and always has pure-strategy equilibria with strong performance guarantees.

\subsection{Utility-Target Auctions}A utility-target auction is a sealed-bid pay-your-bid auction with a special bidding language. A player's bid specifies payments using two pieces of information: her valuation function and the amount of utility she requests (a single real number). Her payment for an outcome is her (claimed) valuation for that outcome minus the utility that she specified in her bid. Formally:
\begin{definition}
A {\em utility-target auction} for a finite outcome space $\outcomes$ is defined as follows:
\begin{itemize}
\item A bid is a tuple $b_i=(x_i,\utiltarget_i)$ where $x\in V_i$ is a function mapping outcomes $\outcome\in\outcomes$ to nonnegative values and $\utiltarget$ is a real number. We call the parameters $x_i$ and $\utiltarget_i$ the {\em value bid} and {\em utility-target bid} respectively.
\item A bidder's effective bid for outcome $o$ is
\[b_i(o)=\max(x_i(o)-\utiltarget_i,0)\enspace.\]
Note this may generate $b_i\not\in V_i$ when the set $V_i$ is sufficiently restricted.
\item The auctioneer chooses the outcome $\outcome^*\in\outcomes$ that maximizes $\sum_{i\in\bidders}b_i(o)$. Ties are broken in favor of the most-recent winning outcome when applicable.
\item When the outcome is $o$, bidder $i$ pays $p_i(o)=b_i(o)$ and derives utility $u_i(o)=v_i(o)-b_i(o)$.  Note that if a bidder reports $x_i=v_i$, then $u_i(o)=\utiltarget_i$ whenever $v_i(o)\geq \utiltarget_i$.
\end{itemize}
\end{definition}
A generic utility-target auction is illustrated in Algorithm~\ref{alg:sa}.

\begin{algorithm}[tb]\label{alg:sa}
\SetKwInOut{Input}{input}\SetKwInOut{Output}{output}
\SetKwIF{WP}{ElseWP}{Otherwise}{with probability}{}{with probability}{otherwise}{end}
\SetKw{KwSet}{Set}
\SetKwRepeat{DoUntil}{do}{until}
\Input{Players' bids $b_i=(x_i,\utiltarget_i)$}
\Output{An outcome $o^*$ and first-price payments $p_i$.}

\BlankLine
\nl Let $b_i(o)=\max(0,x_i(o)-\utiltarget_i)$\tcp*[r]{$b_i(o)$ is $i$'s effective bid for outcome $o$.}
\nl Compute $o^*=\argmax_o\sum_{i\in\bidders}b_i(o)$\tcp*[r]{Choose the outcome with the highest total bid.}
\nl For all $i$, set $p_i=b_i(o^*)$\tcp*[r]{Each player pays what she bid.}
\caption{A generic utility-target auction.}
\end{algorithm}

\section{Quasi-Truthful Bidding}

An idealist's intuition for the utility-target auction is that players truthfully reveal their valuation function through their value bids (i.e. they bid bid $x_i=v_i$) and then use the utility-target bid $\utiltarget_i$ to strategize. Clearly, bidders need not follow this ideal; however, it turns out that they have no incentive to do otherwise --- the utility-target auction is quasi-truthful in the sense that for any bid a player might consider, there is another bid in which she reveals $v_i$ truthfully and obtains at least as much utility:
\begin{lemma}[Quasi-Truthfulness]\label{lem:quasitruthful} Fix the total bid of players $j\neq i$ for all outcomes, i.e. fix $\sum_{j\in[n]\setminus\{i\}}b_j(o)$ for all $\outcome$, and suppose ties are broken according to a fixed total-ordering on outcomes. If bidder $i$ gets $u_i^I$ by bidding $(x_i^I,\utiltarget_i^I)$, then she gets the same utility $u_i^I$ by bidding $(v_i,u_i^I)$.

\end{lemma}

Significantly, this implies bidder $i$ always has a quasi-truthful best-response.

\begin{proof}
Since ties are broken according to a fixed total ordering, the outcome is fully specified by the total bids for each outcome (i.e. by $\sum_{i\in[n]}b_i(\outcome)$ for all $\outcome$). Thus, given $\sum_{j\in[n]\setminus\{i\}}b_j(o)$ and a bid $b_i^I=(x_i^I,\utiltarget_i^I)$ for $i$, the outcome $\outcome^I$ is uniquely defined. Let $\utiltarget_i^I$ be the utility $i$ gets by bidding $(x_i^I,\utiltarget_i^I)$, i.e.
\[u_i^I=v_i(o^I)-b_i(o^I)=v_i(o^I)-\max(x_i^I(o^I)-\utiltarget_i^I,0)\enspace.\]
Now suppose $i$ bids $b_i^Q=(v_i,u_i^I)$ instead of $(x_i^I,\utiltarget_i^I)$. There are two possible results of this change:
\begin{itemize}
\item {\em The outcome doesn't change.} If the outcome doesn't change, then $i$ gets the same utility by construction.
\item {\em The outcome changes to $o^Q\neq o^I$.} Notice that $i$ did not change the amount bid for outcome $o^I$, so the total bid for $o^I$ did not change. Given this and the tie-breaking rule, the only way the outcome can switch from $o^I$ to $o^Q$ is if the total bid for $o^Q$ strictly increased. Given that $\sum_{j\in[n]\setminus\{i\}}b_j(\outcome^Q)$ is fixed, this implies $i$'s bid for $\outcome^Q$ increased, i.e. $b_i^Q(o^Q)>b_i^I(o^Q)\geq0$.

Next, by definition of a utility-target auction, $b_i(\outcome)>0$ implies $x_i(\outcome)>\utiltarget_i(o)$. Since $b_i^Q(o^Q)\geq 0$, this implies $v_i(o^Q)>\utiltarget_i^Q$, from which it immediately follows that $i$'s final utility in $o^Q$ will be $\utiltarget_i^Q$.
\end{itemize}
In either case, $i$'s final utility is precisely $u_i^I$, so $i$ is indifferent between bidding $(x_i^I,\utiltarget_i^I)$ and $(v_i,u_i^I)$.
\end{proof}


\section{Static Equilibrium Analysis}\label{sec:eq}

We begin by studying the utility-target auction from a static perspective and show that they offer strong revenue and welfare guarantees. First, we show that pure-strategy equilibria always exist:
\begin{theorem}\label{thm:cef-exist}
A utility-target auction with $n$ outcomes always has a pure-strategy cooperatively envy-free (defined below) equilibrium that is computable in time $poly(n)$.
\end{theorem}
Specifically, the egalitarian equilibrium exists and is efficiently computable by Algorithm~\ref{alg:s-egal} (proof omitted).

Next, we show that such cooperatively envy-free equilibria not only maximize welfare but offer as much revenue as the VCG mechanism as well as a new revenue benchmark we call the second-price threat (defined below):
\begin{theorem}\label{thm:cef-props}
Any cooperatively envy-free equilibrium of a utility-target auction
\begin{enumerate}
\item maximizes social welfare,
\item dominates the revenue of the VCG mechanism,
\item and has revenue lower-bounded by the second-price threat.
\end{enumerate}
\end{theorem}
We formalize and prove the theorem below.

\subsection{Cooperatively Envy-Free Equilibria}

While a typical utility-target auction may have many equilibria, some of them are unrealistic in repeated auctions. In particular, it is possible to have an equilibrium in which a group of ``losers'' envy the ``winners'' --- the losers would be happy to collectively raise their bids to make an alternate outcome win, but the outcome is an equilibrium because no single bidder is willing to raise her bid high enough. In a repeated setting, one would expect all the losers to eventually raise their bids.

For example, consider a setting with three bidders $(A, B, C)$ and three outcomes $(1, 2, 3)$, in which the first two bidders are symmetric and value the first two outcomes and the third bidder values only the third outcome. Let the specific values, indexed by outcome, be 
\begin{align*}
v_A &= (1, 1.5, 0), \\
v_B &= (1, 1.5, 0),\\
v_D &= (0, 0, 2) .
\end{align*}

Now, let A and B bid for the first outcome, and C bid for the third outcome, with bids: $b_A = (1, 0, 0)$, $b_B =(1, 0, 0)$ and $b_C =(0, 0, 2)$.

Bidders $A$ and $B$ would prefer the second outcome to the first as they see a value of 1.5 instead of 1. Moreover, they would be happy to make the second outcome win by cooperating and each bidding $1+\epsilon$. However, since both are bidding $0$ for the second outcome, neither can unilaterally cause the second outcome to win, making this outcome an equilibrium. The problem in this example is that, at a total price of $2$, bidders $A$ and $B$ would prefer that the second outcome wins. In a sense, bidders $A$ and $B$ in the second outcome envy the deal they received in the first outcome. 

Hence, we are interested in bids such that players have no incentive to cooperatively deviate to get a better outcome. We will call such a set of bids \emph{cooperatively-envy free}. 

To define such a notion, we must also consider bidders who are happy with the winning outcome. Consider a four bidder setting with three possible outcomes, with the following values:
\begin{align*}
v_A &= (1, 1.5, 0),\\
v_B &= (1, 1.5, 0),\\
v_C &= (1, 0.5, 0),\\
v_D &= (0, 0, 2).
\end{align*}

In this case, $C$ cannot get a better deal from the second outcome, so she will not cooperate with $A$ and $B$. In order to win, $A$ and $B$ must collectively bid $1.75$ in the second outcome to make up the deficit between it and the winning outcome (which they are willing to do).
\begin{definition}
The set of bids $\{b_i\}_{i\in \bidders}$ are {\em cooperatively envy-free (CEF)} if there is no subset of bidders $\altbidderset\subseteq \bidders$ who would prefer to cooperatively pay the extra money required to make an alternate outcome $o$ win over the current winner $o^*$.

Formally, a set of bids is cooperatively envy-free if
\[\sum_{i\in\bidders}\max\left((v_i(o)-b_i(o))-(v_i(o^*)-b_i(o^*)),0\right) \leq \sum_{i\in\bidders} b_i(o^*)-b_i(o)\]
for all outcomes $o$.
\end{definition}
The CEF constraints are similar to the core property described by Milgrom~\cite{M04} in package auctions (as well as notions like group-strategyproofness); however, the notion of a CEF outcome is weaker. For example, it does not require that bidders are playing equilibrium strategies.

Equilibria that are CEF have nice properties analogous to those of core equilibria in package auctions. The following claims are straightforward and proven in the appendix of the full version of the paper:
\begin{claim}\label{clm:welfare-max}
CEF bids maximize welfare.
\end{claim}
\begin{claim}\label{clm:vcg-dominance}
The revenue from CEF bids dominate that of the VCG mechanism: every player pays at least as much in the CEF equilibrium as she would in the VCG mechanism.
\end{claim}

\subsection{The ``Second-Price Threat''}

The revenue of a CEF equilibrium also meets or exceeds a benchmark we call the second-price threat. The revenue of the second-price auction has a convenient intuition: the price paid by the winner should be at least as large as the maximum willingness to pay of any other bidder. We can ask the same question in more general settings: how much would ``losers'' be willing to pay to get an outcome $o$ instead of the socially optimal outcome $o^*$? In general, player $i$ should be willing to pay up to $v_i(o)-v_i(o^*)$ to help $o$ beat $o^*$, hence we can generalize the intuition of the second-price auction to give a natural lower bound on the revenue the auctioneer might hope to earn:
\begin{definition}
The {\em second-price threat} for outcome $\optoutcome$ is given by
\[\max_{o\in \outcomes}\sum_{i\in\bidders}\max(v_i(o)-v_i(o^*),0)\enspace.\]
\end{definition}
This bound is particularly powerful in cases where bidders share value for an outcome (cases where VCG would make little or no revenue). For example, consider the following 4-bidder, 2-outcome setting:
\begin{align*}
v_A &= (1, 0)\\
v_B &= (1, 0)\\
v_C &= (1, 0)\\
v_D &= (0, 2)
\end{align*}
In a VCG auction, nobody pays anything. However, a na\"{i}ve auctioneer might expect the first outcome to win, with $A$, $B$, and $C$ paying a total of \$2 (the second-price threat) since they are beating $D$.

The CEF constraints quickly imply that a CEF outcome generates at least as much revenue as the second-price threat:
\begin{claim}\label{clm:second-price-bound}
The revenue in any CEF outcome is lower-bounded by the second-price threat.
\end{claim}
Proof is given in the appendix of the full version of the paper.

\begin{algorithm}[tb]\label{alg:s-egal}
\SetKwInOut{Input}{input}\SetKwInOut{Output}{output}
\SetKwIF{WP}{ElseWP}{Otherwise}{with probability}{}{with probability}{otherwise}{end}
\SetKw{KwSet}{Set}
\SetKwRepeat{DoUntil}{do}{until}
\Input{A utility-target auction problem.}
\Output{The egalitarian equilibrium bids $b_i^*=(v_i,\utiltarget_i^*)$.}
\DontPrintSemicolon
\BlankLine
\nl Set all bids to $(v_i,0)$. Call the socially optimal outcome $o^*$.\;
\nl Increase $\utiltarget_i$ for all bidders uniformly until some bidder $i$ reaches $\utiltarget_i=v_i(o^*)$ or a CEF constraint would be violated for some outcome $o$.\;
\nl Fix the bids of the newly-constrained advertisers.\;
\nl Repeat (2) and (3), lowering only unfixed bids until all bidders are fixed.\;
\caption{An algorithm for computing the egalitarian equilibrium in a utility-target auction.}
\end{algorithm}

\section{Utility-Target Auctions for Sponsored Search}
\label{sec:adauction}

Sponsored search advertising demonstrates the benefits of a utility-target auction. The standard auction in this setting is the generalized second-price (GSP) auction; however, it (and incentive-compatible VCG mechanisms) lack transparency: payments are complicated to compute and bidders must trust the auctioneer not to abuse their knowledge when an auction is repeated. Moreover, its performance may degrade when using more accurate models of user behavior~\cite{RT12} and advertiser value~\cite{HJW}. It can have misaligned incentives when parameters are estimated incorrectly~\cite{WS12}. Some of these problems would be solved by a first-price auction; however, Overture's implementation of GFP demonstrated that such schemes might be highly unstable. A utility-target auction offers the benefits of a pay-your-bid auction without the instability of GFP.

\subsection{The Utility-Target Ad Auction} We illustrate a utility-target auction in the standard model of sponsored search: $n$ advertisers compete for $m\leq n$ slots associated with a fixed keyword. An advertiser's value depends on the likelihood of a click, called the click-through-rate (CTR) $c$, and the value $v$ to the advertiser of a user who clicks. The CTR $c$ is separable into a parameters $\beta_i$ that depends on the advertiser and $\alpha_j$ that depends on the slot, so the expected value to advertiser $i$ for having her ad shown in slot $j$ is $c_{i,j}v_i=\alpha_j\beta_iv_i$. As is standard, we assume that slots are naturally ordered from best ($j=1$) to worst ($j=m$), i.e. $\alpha_{j}\geq \alpha_{j'}$ for all $j<j'$. Without loss of generality, we assume bidders are ordered in decreasing order of bid, i.e. $b_1\geq b_2\geq\dots\geq b_n$.

The auctioneer chooses a matching of advertisements to slots and charges an advertiser a per-click price $ppc_i$. For example, in the GFP auction, advertisers submitted bids $b_i$ representing their per-click payment and paid $ppc_i=b_i$ whenever a their ads were clicked. Similarly, in the standard GSP auction, bidder $i$ is charged according bid of the next highest bidder.\footnote{The designers of the GSP auction intended it to inherit the incentive compatibility of the second-price auction. It does not; however, it has the nice property that bidder $i$ pays the minimum amount required to win the slot that she received.} To account for differences in CTRs, this quantity is normalized by $\beta$ so that bidder $i$ pays a per-click price of $ppc_i=\frac{\beta_i}{\beta_{i+1}}b_{i+1}$.

In a {\em utility-target auction}, bidders submit both their per-click value $x_i$ and the utility-target bid $\utiltarget_i$ (the utility that they request). The auctioneer picks the assignment $j(i)$ maximizing
\[\sum_{i\in\bidders}\max(0,\alpha_{j(i)}\beta_ix_i-\utiltarget_i)\]
and charges $i$ so that her expected payment is
\[\ex[p_i]=\max(0,\alpha_{j(i)}\beta_ix_i-\utiltarget_i)\enspace.\]

There are at least two interesting ways the utility-target auction can be implemented. The first implementation charges
\[ppc_i=\max\left(0,x_i-\frac{\utiltarget_i}{\alpha_{j(i)}\beta_i}\right)\]
to achieve the desired expected payment. In effect, it uses the utility request $\utiltarget_i$ to compute a different per-click bid for each slot. A practical downside to this implementation is that the payments are still somewhat complicated from the bidders' perspectives; however, the auctioneer could mitigate this problem by publishing CTRs and displaying the per-click payments in the bidding interface.

An alternative implementation of the utility-target auction pays a rebate of $\utiltarget_i$ regardless of whether a click occurred and charges precisely $ppc_i=x_i$ when a click occurs. This auction is even simpler from the bidders' perspective; however, when a click does not occur the auctioneer will be paying the bidder (in expectation the bidder still pays the auctioneer). This implementation of the utility-target auction is illustrated in Algorithm~\ref{alg:aa}.

Such a utility-target auction offers many benefits over existing auction designs like GSP and VCG. As noted earlier, a first-price auction directly increases transparency and simplicity from the bidders' perspective. Even if bidders reveal their true valuation functions $v_i$, the pay-your-bid property ensures that increasing a reserve price will not increase payments unless bidders subsequently raise their bids.

The auction also easily generalizes to more complicated bidding languages. Whereas the welfare and revenue performance of GSP degrades (albeit gracefully)~\cite{RT12} when considering externalities imposed by the presence of competing ads, the reasonable (CEF) equilibria of the utility-target auction guarantee good performance. An open question of~\cite{HJW} is to find an auction that performs well in multi-slot settings when multiple bidders can benefit from clicks on the same ad (e.g. Microsoft and Samsung both benefit from an ad for a Samsung laptop running Windows) --- the utility-target auction offers good revenue guarantees in these `coopetitive' ad auctions with multiple slots. The utility-target auction is also less sensitive to estimation errors in the CTRs. As shown in~\cite{WS12}, incentive-compatibility can be broken because the auctioneer only knows estimates of the $\alpha$ and $\beta$ parameters. Informally, the utility-target auction is much less sensitive to such errors because the payments need not explicitly depend on the auctioneer's estimates.

\subsection{Utility-Target vs. GFP} Juxtaposing GFP with the utility-target auction illustrates the benefits of a more complex bidding language. GFP is identical to the utility-target ad auction except that bids contain only the per-click payment $x_i$ and not the utility-target bid $\utiltarget_i$. Consequently, a player's bid necessarily offers the same per-click payment regardless of the slot won by the bidder. By comparison, the utility-target auction permits bids that encode a different per-click payment depending on the slot in which an ad is shown.

In retrospect, it is easy to see that different per-click bids are important for a good pure-strategy equilibrium. In GFP, all advertisers who are shown must bid so that $\beta_ix_i$ is the same, otherwise some bidder can lower her value of $x_i$ without changing her assignment; however, if this is true, then some bidder can move up to the top slot by bidding $x_i+\epsilon$. In fact, any bidding language that requires the same per-click payment for all slots could not have a pure-strategy equilibrium unless the potential benefit of being in the top slot was less than the effective bid increment required to get there. This necessarily weakens any revenue guarantees and, worse, implies that the auction cannot differentiate between the winning bidders to pick the best ordering of ads.

As noted earlier, the existence of a pure-strategy equilibrium is directly related to the dynamic performance of the auction. Edelman and Ostrovsky~\cite{EO07} discuss how the lack of such an equilibrium naturally leads to sawtooth cycling behavior in GFP, as bidders alternate between increasing their bids to compete for higher slots and decreasing their bids to avoid overpaying for the slots they have. They also show that this cyclic behavior potentially reduced revenue below that of the VCG mechanism. In contrast, Theorem~\ref{thm:cef-exist} shows that utility-target auctions have pure-strategy equilibria, and Theorem~\ref{thm:cef-props} shows that revenue at equilibrium dominates the VCG mechanism; moreover, our dynamic results show that bids will naturally approach this equilibrium (or the set of such equilibria) as bidders adjust their utility targets.

\begin{algorithm}[tb]\label{alg:aa}
\SetKwInOut{Input}{input}\SetKwInOut{Output}{output}
\SetKwIF{WP}{ElseWP}{Otherwise}{with probability}{}{with probability}{otherwise}{end}
\SetKw{KwSet}{Set}
\SetKwRepeat{DoUntil}{do}{until}
\Input{Bids $b_i=(x_i,\utiltarget_i)$.}
\Output{An assignment of advertisements to slots, per-click payments $ppc_i$, and unconditional payments $r_i$.}
\BlankLine
\nl For each bidder $i$ and slot $j$, compute $\ex[p_{i,j}]=\max(\alpha_j\beta_ix_i-\utiltarget_i,0)$\;
\nl Compute assignment $j(i)$ of advertisements to slots that maximizes $\sum_{i\in\bidders}\ex[p_{i,j(i)}]$\;
\nl \eIf{$\ex[p_{i,j(i)}]>0$}{
Always pay $i$ the rebate $r_i=\utiltarget_i$\;
Whenever $i$'s ad is clicked, charge $ppc_i=x_i$\;}
{Do not charge/pay anything to $i$\;}
\caption{Autility-target auction for search advertising.}
\end{algorithm}


\section{Dynamic Analysis}
In this section, we consider the behavior of utility-target auctions in a very simple dynamic setting, and under very simple assumptions. We show that a few rules and simple knowledge of whether one is winning or losing are enough to guarantee the revenue and welfare bounds from Theorem~\ref{thm:cef-props}.

Following Lemma~\ref{lem:quasitruthful}, we assume that bidders are quasi-truthful and report bids of the form $(v_i,\utiltarget_i)$. Bidders compete using the utility-target terms $\sbid{i}$ and employ strategies to optimize their utility.

\paragraph{Winners and Losers\nonstoc{.}} Our dynamic axioms are based on a natural decomposition of bidders into winners and losers. In a standard first-price auction, the winner is the bidder who gets what he wants --- the item --- and the losers are those who do not get what they want. In a utility-target auction, a bidder effectively reports her valuation $v_i$ and requests that the auctioneer give her a certain utility $\sbid{i}$. This suggests partitioning bidders into winners and losers based on whether a bidder gets the utility she requests, giving us the following formal definition:
\begin{definition}[Winners and Losers]
A {\em winner} is a bidder who gets the utility she requests, i.e. if $o^b$ is the outcome of the auction, then $i$ is a winner if and only if
\[u_i(o^b)=v_i(o^b)-b_i(o^b)=\utiltarget_i\enspace.\]
Any bidder who is not a winner is a {\em loser}.  
\end{definition}
\begin{observation}
Bidder $i$ is a winner if and only if $v_i(o^b)\geq \utiltarget_i$ and is always a winner when $\utiltarget_i=0$. When bidder $i$ is a loser, $u_i(o^b)<\utiltarget_i$.
\end{observation}
Note that this definition does not coincide with the standard definition of winners and losers in a single item auction because a bidder who does not get the item is still a winner if $\utiltarget_i=0$.

\paragraph{Raising and Lowering Bids\nonstoc{.}} In a dynamic setting, we want to think about how winners and losers manipulate $\utiltarget_i$. In the utility-target auction, the effective bid $b_i$ (what bidder $i$ is actually offering to pay) and the utility-target term $\utiltarget_i$ move in opposite directions, so when we talk about raising $i$'s bid we are talking about decreasing the utility-target term $\utiltarget_i$:
\begin{definition}[Raising and Lowering Bids]
We say that bidder $i$ {\em raises her bid} from $(v_i,\utiltarget_i)$ if she chooses a new bid $(v_i,\utiltarget_i')$ where $\utiltarget_i'< \utiltarget_i$, i.e. she raises her bid if she {\em decreases her utility-target bid}.

Similarly, a bidder who {\em lowers her bid} correspondingly {\em increases her utility-target bid} from $\utiltarget_i$ to $\utiltarget_i'>\utiltarget_i$.
\end{definition}
Importantly, our definition of winners and losers shares a natural property with the standard definition: winners cannot benefit by offering to pay more, and losers cannot benefit by offering to pay less:
\begin{claim}\label{clm:winlower}
Fixing other players' bids, a loser cannot increase her utility by raising her bid. Likewise, a winner cannot increase her utility by lowering her bid.
\end{claim}
The claim is straightforward to prove.

Our definition of winners and losers also has a new property that is important:
\begin{claim}\label{clm:losealwaysraise}
A loser can always raise her bid in a way that weakly increases her utility.
\end{claim}
\begin{proof} Suppose $i$ is a loser bidding $(v_i,\utiltarget_i)$ and receiving utility $u_i<p_i$. If she raises her bid to $(v_i,u_i)$, Lemma~\ref{lem:quasitruthful} says that she will receive utility of precisely $u_i$, making her a winner.\end{proof}

In our model, bidders locally adjust their bids by $\epsilon$. To mimic settings where auctions happen frequently and no two bidders move simultaneously, bid changes are modeled as asynchronous events. As noted earlier, our model assumes players bid quasi-truthfully, that is, they always submit their true valuation functions in their bids. As a result, the history of the auction is characterized by a sequence of utility-target vectors $\utiltarget^0,\dots$.

We assume that $0\leq \inf_o v_i(o)$ and $\sup v_i(o)<\infty$, so utility-targets will always lie in the finite interval $[0,\sup v_i(o)]$. Unless a player's utility-target hits the boundary of this interval, all bid changes are made in increments of $\epsilon$.

\paragraph{Notions of Convergence\nonstoc{.}} We will show that progressively stronger assumptions imply progressively stronger convergence guarantees. Our first results show that bids will eventually be close to the set of CEF (or non-CEF) bids. As noted earlier, the utility-targets $\allsbid$ are sufficient to characterize bidders' strategies, so we define $\cefset$ to be the set of all such utility-target:
\begin{definition}[The CEF Set]
\label{def:allcefbids}
$\cefset$ is the set of all utility-target vectors $\allsbid$ where the quasi-truthful bids $(v_i,\utiltarget_i)$ produce a cooperatively envy-free outcome.\footnote{Note that membership in $\cefset$ depends on both the vector $\utiltarget$ and the outcome chosen by the auction. This is because certain utility-target vectors $\utiltarget$ will be in $\cefset$ if ties are broken in favor of $o^*$ but not if ties are broken in favor of a suboptimal outcome.}

The set $\ncefset$ is the set of all utility-target vectors which are {\em not} CEF, i.e. $\ncefset=\Re_+^n\setminus\cefset$.
\end{definition}
Significantly, $\cefset$ is never empty. In particular, it always contains the 0 vector ($0^n\in\cefset$).

Since bidders are continually experimenting with their bids, it is not realistic to expect bids to explicitly converge to $\cefset$; rather, they will remain close. For a set of bids $\utiltarget$, let $\utiltarget_\epsilon$ denote the set of bids that are close to some vector in $\utiltarget$, i.e.
\begin{definition}
\label{def:nearcefbids}
Let $S_\epsilon$ be the set of all utility-targets $\utiltarget$ which are close to some vector in $\utiltarget$ coordinate-wise. Formally,
\[S_\epsilon = \{\allsbid\  |\ \exists \utiltarget'\in S\mbox{ s.t. }||\utiltarget-\utiltarget'||_\infty\leq\epsilon\}\enspace.\]
\end{definition}
In particular, we will care about the sets $\nearcef$ and $\nearncef$, the sets representing bids close to being CEF and close to being not CEF, respectively.

Next we define the convergence of an auction to utility-target bids $\utiltarget$:
\begin{definition}
An auction {\em converges to a set of utility-targets $S$} if, for any $\delta>0$, there exists a sufficiently small bid adjustment parameter $\epsilon$ for which the auction always reaches a utility-target $\utiltarget$ such that all future bids are in $S_\delta$.
\end{definition}

Our strongest result wull show that bids converge to the egalitarian equilibrium:

\begin{definition}
The egalitarian equilibrium is the CEF equilibrium which distributes utility as evenly as possible. Formally, for each equilibrium let $u_{\uparrow}$ be the vector of bidders' utilities with its coordinates sorted in increasing order. The {\em egalitarian equilibrium} is the one for which $u_{\uparrow}$ is lexicographically maximized.
\end{definition}
An auction converges to the egalitarian equilibrium $\utiltarget^{E}$ if it converges to $\{\utiltarget^{E}\}$.

\subsection{Axioms and Results}

Our convergence theorems show that progressively stronger assumptions about bidder behavior lead to progressively stronger convergence results.

Our first axiom of bidder behavior captures some intuition about how winners and losers behave. Following Claim~\ref{clm:winlower}, a winner cannot benefit by raising her bid and a loser cannot benefit by lowering it, so we suppose that they never do this. Additionally, a loser who is actively engaged in the auction should raise her bid if it is beneficial. By Claim~\ref{clm:losealwaysraise} we know that a loser can always raise her bid in a way that is weakly beneficial, so we suppose that a loser will always try to raise her bid.

\begin{itemize}

\item[(\Aa).] {\em A losing bidder will raise her bid in an effort to win; a loser will not lower her bid and a winner will not raise her bid.} Formally, if the current utility-target is $\allsbid$ and $i$ is a loser, then $i$ must raise her bid at some point in the future unless she becomes a winner through the actions of other bidders.

\end{itemize}

Anecdotal evidence suggests that advertisers bidding in an ad auction generally expend substantial effort to launch advertising campaigns but are much slower to change things once they appear to work. Our second axiom generalizes this idea by supposing that winners (who, by definition, get the utility-target they request) view the outcome of the auction as a success while losers are unhappy with the results:

\begin{itemize}
\item[(\Ab).] {\em A bidder who is losing is more impatient than a bidder who is winning.} Formally, if the current utility-target is $\allsbid$ and a set of bidders $L\subseteq[n]$ are losers, then the next time bids change it will necessarily be because some loser $i\in L$ raised her bid.

\end{itemize}

Our third axiom is analogous to (\Aa) but for winners --- a winner who is actively engaged should lower her bid from time to time to see if she can win at a lower bid.
\begin{itemize}
\item[(\Ac).] {\em A winner will try lowering her effective bid to win at a lower price. Specifically, if a bidder is currently a winner, then she must lower her bid at some point in the future unless she becomes a loser through the action of another player.} Formally, if the current utility-targets are $\allsbid$ and $i$ is a winner, then $i$ must lower her bid at some point in the future unless she becomes a loser through the actions of other bidders.

\end{itemize}

Our final axiom concerns the relative timing of events. Intuition and anecdotal evidence suggests that larger bidders who have more at stake tend to invest more heavily in active bidding strategies. This axiom roughly represents that intuition:
\begin{itemize}
\item[(\Ad).] {\em Between two losers, the bidder with the higher utility-target is more impatient.} Formally, if the current utility-targets are $\allsbid$ and bidders $i$ and $j$ are both losers, then $i$ will raise her bid before $j$ if $\utiltarget_i > \utiltarget_j$.

\end{itemize}

These simple properties of bidder behavior imply the following convergence results. Proofs follow in Section \ref{sec:dynamicsbehavior} and Appendix \ref{sec:sel-proofs}.
\begin{theorem}
\label{thm:converge-sink}
If losing bidders will only raise their effective bids (\Aa) and are more impatient than winning bidders (\Ab), the auction converges to the set of bids that are cooperatively envy-free (i.e. bids will be in $\nearcef$).
\end{theorem}

\begin{theorem}
\label{thm:converge-out-sink}
If winners try to lower their effective bids (\Ac) and losers try to raise but not lower their effective bids (\Aa), the auction converges to the set of bids that are non-cooperatively envy-free (i.e. bids will be in $\nearncef$).
\end{theorem}

Combining Theorems \ref{thm:converge-sink} and \ref{thm:converge-out-sink} shows that bids will converge to the frontier of the CEF set. The strict Pareto frontier of this set is the set of CEF equilibrium bids.
\begin{corollary}
If losing bidders will try raising their bids (\Aa), losers are less patient than winners (\Ab), and winners try lowering their bids (\Ac), the auction converges to the boundary between CEF and non-CEF bids (bids will be in the set $\nearcef\cup\nearncef$).
\end{corollary}

Finally, adding \Ad{} induces convergence to a particular equilibrium:
\begin{theorem}
\label{thm:converge-egal}
If losing bidders will raise their effective bids (\Aa), winning bidders will try lowering their effective bids (\Ac), and the most impatient bidder is the losing bidder bidding for the highest utility-target (\Ab, \Ad), then bids will converge to the Egalitarian envy-free equilibrium.
\end{theorem}

\subsection{Convergence Proofs}
\label{sec:dynamicsbehavior}

In this section we give proofs of Theorems \ref{thm:converge-sink} and \ref{thm:converge-out-sink}. Theorem \ref{thm:converge-egal} and omitted proofs may be found in Appendix \ref{sec:sel-proofs}. 
Throughout this section, we assume that there is a single welfare optimal outcome for clarity of presentation.

\begin{observation}\label{obs:lower-allwin}
Under assumptions \Aa{} and \Ab, a bidder will only lower her bid if all bidders are winners.
\end{observation}

\begin{lemma}
\label{lem:allwin-cef} If all bidders are winners under utility-targets $\allsbid$, then $\allsbid$ is in the CEF set $\cefset$.
\end{lemma}
\begin{proof}
If all bidders are winners, then we know that they are receiving precisely the utility-target they request when they bid $\sbid{}$. Intuitively, this means that raising bids necessarily implies receiving less utility.

Formally, if bids are $b_i=(v_i,\utiltarget_i)$ and the outcome of the auction is $o^b$, then we want to show that the CEF condition holds for any outcome $o$. Since all bidders are winners, we know $v_i(o^b)-b_i(o^b)= \utiltarget_i$. Moreover, $v_i(o)-b_i(o)\leq \utiltarget_i$ by definition, so $v_i(o)-b_i(o)\leq v_i(o^b)-b_i(o^b)$. Thus
\[\max\left((v_i(o)-b_i(o))-(v_i(o^b)-b_i(o^b)),0\right)=0\enspace.\]

Since $o^b$ is the outcome of the auction, we know $\sum_{i\in\bidders}b_i(o^b)\geq\sum_{i\in\bidders}b_i(o)$ for any outcome $o$. Thus, $0\leq\sum_{i\in\bidders}b_i(o^b)-b_i(o)$ and therefore
\[\sum_{i\in\bidders}\max\left((v_i(o)-b_i(o))-(v_i(o^b)-b_i(o^b)),0\right)\leq\sum_{i\in\bidders}b_i(o^b)-b_i(o)\]
as desired.
\end{proof}
Since \Aa{} and \Ab{} imply that a player will only lower her bid from $\sbid{}$ if all bidders are winners, an important corollary is that a bidder will only lower her bid if the current utility-target vector is in the CEF set $\cefset$:
\begin{corollary}
\label{cor:lowering-cef} Under assumptions \Aa{} and \Ab, if a player lowers her bid from $\sbid{}$, then $\sbid{}$ is in the CEF set $\cefset$.
\end{corollary}
A corollary of Claim~\ref{clm:welfare-max} is that any set of CEF bids maximizes welfare, hence this implies that a player will only lower her bid if the welfare-optimal outcome is winning:
\begin{corollary}
\label{cor:lowering}
Under assumptions \Aa\, and \Ab, a bidder will only lower her bid if a welfare-optimal outcome $\optoutcome$ is winning.
\end{corollary}

Another useful fact about $\cefset$ is that it is leftward-closed (the proof is in the appendix) and the natural corollary that $\ncefset$ is rightward-closed:
\begin{lemma}\label{lem:cefconvex}
If $\utiltarget$ is in the CEF set $\cefset$, then $\utiltarget-\Delta$ is in the CEF set $\cefset$ for any $\utiltarget\geq\Delta\geq0$.
\end{lemma}
\begin{corollary}\label{cor:ncefconvex}
If $\utiltarget$ is in the not-CEF set $\ncefset$, then $\utiltarget+\Delta$ is in the not-CEF set $\ncefset$ for $\Delta\geq0$.
\end{corollary}

To prove that bids will converge, we first show that bids will not be stuck at arbitrarily low values:
\begin{lemma}\label{lem:allwin-finite}
Suppose the initial vector of utility-targets is $\utiltarget^0$. Under properties \Aa{} and \Ab, the auction will always reach a configuration in which all bidders are winners and will do so within $\left|\left\lceil\frac1\epsilon \utiltarget^0\right\rceil\right|_1$ steps.
\end{lemma}

We can now prove the our first theorem, that bids will be close to $\cefset$ when \Aa{} and \Ab{} are satisfied.

\begin{proofof}{Theorem~\ref{thm:converge-sink}}
Lemma~\ref{lem:allwin-finite} implies that all bidders will be winners within a finite time. Once all bidders are winners, the only way bids will change is if someone lowers her bid. Thus, after a finite amount of time, we can conclude that either all bidders are winners or some bidder has lowered her bid.

Let $\utiltarget$ be the vector of utility-targets at any point after the first time all bidders are winners. If all bidders are still winners, then $\utiltarget\in\cefset$ by Lemma~\ref{lem:allwin-cef}. Otherwise, let $i$ be the most recent player to lower her bid, increasing the utility-target vector from $\allsbid'$ to $\allsbid''=\utiltarget'+\epsilon e_i$. We show that if $i$ raises her bid again then the resulting utility-targets must be CEF regardless of how bids have changed since $i$'s raise.

By construction, players have only raised their bids since $i$ lowered hers, so we can define $\Delta=\allsbid''-\allsbid$ where $\Delta\geq0$. Corollary~\ref{cor:lowering-cef} tells us that $\utiltarget'\in\cefset$. If $i$ raised her bid between $\utiltarget''$ and $\utiltarget$, then $\utiltarget\leq \utiltarget'$ and Lemma~\ref{lem:cefconvex} tells us that $\utiltarget\in\cefset$, so were done. Otherwise, we know $\allsbid''\geq\Delta\geq0$ and Lemma~\ref{lem:cefconvex} tell us that $\allsbid'-\Delta\in\cefset$. Therefore  $\allsbid=\allsbid'-\Delta+\epsilon e_j\in\nearcef$.
\end{proofof}

To prove Theorem~\ref{thm:converge-out-sink}, we need a lemma similar to Lemma~\ref{lem:allwin-finite} showing that the auction will reach a bid vector that is CEF:

\begin{lemma}\label{lem:notallwin-finite}
Under properties \Aa{} and \Ac, as long as there is some outcome $o$ and bidder $j$ such that $v_j(o)>v_j(o^*)$, the auction will always reach a configuration that is not CEF when $\epsilon$ is sufficiently small. If there is no such outcome $o$ and bidder $j$, then the auction may converge to $\utiltarget_j=v_j(o^*)$ instead.
\end{lemma}

\begin{proofsketchof}{Theorem~\ref{thm:converge-out-sink}}
By Lemma~\ref{lem:notallwin-finite}, the auction will typically reach a utility-target vector in $\ncefset$. Our primary goal is to show that the auction will be in $\nearncef$ from that point onwards.

Suppose $\utiltarget\in\cefset$. Let $\utiltarget'$ be the most recent utility-targets that were in $\ncefset$ and let $\utiltarget''\in\cefset$ be the utility-targets immediately after $\utiltarget'$. Let $i$ be the bidder who changed her bid between $\utiltarget'$ and $\utiltarget''$. Let $o$ be an outcome that violated the CEF constraints at $\utiltarget'$.

Consider a bidder $j$ whose bid is higher at $\utiltarget$ than at $\utiltarget''$. Notice that $j$ must be a loser to raise her bid, roughly $v_j(o^*)<\utiltarget_j$. It then follows from the definition of the utility-target auction that $b_j(o^*)$ essentially did not change from $\utiltarget'$ to $\utiltarget$. Moreover, $b_j(o)$ can only have increased from $\utiltarget'$ to $\utiltarget$, so $b_j(o)$ increased more than $b_j(o^*)$ did. Similar logic leads to an analogous conclusion for bidders whose bids were lower at $\utiltarget''$ than $\utiltarget$, roughly giving
\[b_j(o)-b_j(o^*)\geq b_j'(o)-b_j'(o^*)\]
for each bidder $j\neq i$. For bidder $i$, a similar inequality holds accounting for the fact that $i$ raised her bid from $\utiltarget'$ to $\utiltarget''$:
\[b_i(o)-b_i(o^*)\geq b_i'(o)-b_i'(o^*)-\epsilon\enspace.\]
Finally, given that the CEF constraint for $o$ was violated at $\utiltarget'$, these inequalities imply it must be nearly violated at $\utiltarget$.
\end{proofsketchof}

We have now shown that when losing bidders raise their effective bids and winning bidders lower their effective bids, bids remain close to the frontier of the CEF set $\cefset$. Adding in the specific behavior that the first player to raise their bid will be the losing bidder with the highest utility-target results in convergence to one specific equilibrium: the egalitarian equilibrium (Theorem \ref{thm:converge-egal}). The full proof is included in the appendix; we provide a sketch of it here.

\begin{proofsketchof}{Theorem~\ref{thm:converge-egal}}

Arrange bidders into levels $L_1, \ldots, L_k$ in increasing order of the utility each bidder gets at the egalitarian equilibrium.

For each level $\level_{i+1}$, bids from all bidders in the level will converge close to the egalitarian equilibrium once the bids of lower level bidders are sufficiently close to their egalitarian bids.

Thus, beginning with the bidders who get the least utility in equilibrium, and working on up to the lucky bidders with the most utility, bids will converge close to the egalitarian outcome.
\end{proofsketchof}

\section{Conclusion and Open Questions}

Pay-your-bid auctions --- and utility-target auctions in particular --- offer many advantages over incentive compatible mechanisms in terms of transparency and simplicity. Moreover, in many complex settings they even appear to generate more revenue.

Our work first shows that the bidding language is important in first-price auction design. In particular, it is both important and sufficient that bidders can compete in terms of their final utility. Also, a key feature in a repeated first-price auction is a pure-strategy equilibrium, something that GFP does not have~\cite{EOS07}. This is a question of design: the existence of pure-strategy equilibria may be guaranteed through a carefully crafted bidding language (e.g. the utility-target auction) that can encode different per-click payments for different ad slots.

More significantly, when players compete on utility, our results show that robust performance guarantees may be derived using only simple axioms of bidder behavior that merely require knowledge of whether one is winning or losing. These results are powerful because they do not require an a priori assumption that the auction is in equilibrium or full information about others' bids.

Yet, reflection raises a concern about utility-target auctions: {\em why should bidders reveal their true valuation functions in a repeated auction}? We claimed that first-price auctions were better because the auctioneer could not cheat, but it would seem that quasi-truthfulness is just as dangerous. In fact, a quasi-truthful pay-your-bid auction is still strongly preferable to a standard second-price auction: even if the auctioneer knows a bidder's true valuation function, it cannot immediately increase the amount of money the bidder pays. By comparison, the auctioneer in a second-price auction might force a bidder to pay her full value in the second round by increasing the reserve price. The auctioneer is welcome to engage in a game of chicken or a ``negotiation'' with the bidder to see if she is willing to raise her bid, but the pay-your-bid property ensures that final approval still rests with the bidder.

In practice, systems may also be designed to encourage competition on the utility-target term and thereby recover stability. For example, Overture exacerbated the instability of the GFP auction by offering an API automating the sawtooth behavior. If an API were offered to compete on the utility-target term, bidders would likely use the API and stability would be restored, regardless of whether they were reporting their true valuation functions.

Issues of quasi-truthfulness aside, our work also raises questions about dynamic axioms of bidder behavior. Our axioms may be simple and natural, but strict adherence to them is clearly unrealistic. In this vein, many interesting questions are open:
\begin{enumerate}
\item {\em How does the behavior of the auction change with small modifications to the axioms?} For example, we showed that bids would converge to the egalitarian equilibrium when the bidder with the most to gain raised first. Can we prove convergence to a different equilibrium by modifying players' delays? 
\item {\em Do the performance guarantees still hold if axioms only hold probabilistically or on average?} It seems unlikely that bidder behavior always satisfies any particular set of axioms. How do the dynamic guarantees change when axioms only hold most of the time?
\item {\em What dynamic axioms do bidders actually obey?} An interesting experimental question is to determine what axioms are actually satisfied by bidder behavior. For example, could one experimentally measure bidders' delays and combine this with an answer to (1) to predict a particular equilibrium outcome?
\end{enumerate}

\bibliographystyle{plain}
\bibliography{../../miaa}

\begin{thebibliography}{10}

\bibitem{ABH59}
Kenneth~J. Arrow, H.~D. Block, and Leonid Hurwicz.
\newblock On the stability of the competitive equilibrium, ii.
\newblock {\em Econometrica}, 27(1):82--109, 1959.

\bibitem{ABHLT10}
Itai Ashlagi, Mark Braverman, Avinatan Hassidim, Ron Lavi, and Moshe
  Tennenholtz.
\newblock Position auctions with budgets: Existence and uniqueness. working
  paper, 2010.

\bibitem{AM06}
Lawrence~M. Ausubel and Paul Milgrom.
\newblock The lovely but lonely vickrey auction.
\newblock In {\em Combinatorial Auctions, chapter 1}. MIT Press, 2006.

\bibitem{BM92}
Venkatesh Bala and Mukul Majumdar.
\newblock Chaotic tatonnement.
\newblock {\em Economic Theory}, 2(4):437--45, October 1992.

\bibitem{BW86}
B.~Douglas Bernheim and Michael~D. Whinston.
\newblock {Menu Auctions, Resource Allocation, and Economic Influence}.
\newblock {\em The Quarterly Journal of Economics}, CI(1), 1986.

\bibitem{CF08}
Richard Cole and Lisa Fleischer.
\newblock Fast-converging tatonnement algorithms for one-time and ongoing
  market problems.
\newblock In {\em STOC '08: Proceedings of the 40th annual ACM symposium on
  Theory of computing}, pages 315--324, New York, NY, USA, 2008. ACM.

\bibitem{DM07}
Robert Day and Paul Milgrom.
\newblock {Core-selecting package auctions}.
\newblock {\em International Journal of Game Theory}, 36(3-4):393--407, July
  2007.

\bibitem{EO07}
Benjamin Edelman and Michael Ostrovsky.
\newblock Strategic bidder behavior in sponsored search auctions.
\newblock {\em Decision Support Systems}, 43(1):192--198, 2007.

\bibitem{EOS07}
Benjamin Edelman, Michael Ostrovsky, and Michael Schwarz.
\newblock Internet advertising and the generalized second-price auction:
  Selling billions of dollars worth of keywords.
\newblock {\em American Economic Review}, 97(1):242--259, March 2007.

\bibitem{ES10}
Benjamin Edelman and Michael Schwarz.
\newblock Optimal auction design and equilibrium selection in sponsored search
  auctions, 2010.

\bibitem{EK98}
Richard Engelbrecht-Wiggans and Charles~M. Kahn.
\newblock Multi-unit pay-your-bid auctions with variable awards.
\newblock {\em Games and Economic Behavior}, 23(1):25 -- 42, 1998.

\bibitem{EMN09}
Eyal Even-dar, Yishay Mansour, and Uri Nadav.
\newblock On the convergence of regret minimization dynamics in concave games.
\newblock In {\em STOC '09: Proceedings of the 41st annual ACM symposium on
  Theory of computing}, pages 523--532, New York, NY, USA, 2009. ACM.

\bibitem{H89}
Glenn~W. Harrison.
\newblock Theory and misbehavior of first-price auctions.
\newblock {\em The American Economic Review}, 79(4):pp. 749--762, 1989.

\bibitem{HM00}
Sergiu Hart and Andreu Mas-Colell.
\newblock A simple adaptive procedure leading to correlated equilibrium.
\newblock {\em Econometrica}, 68(5):1127--1150, 2000.

\bibitem{HN12}
Sergiu Hart and Noam Nisan.
\newblock The menu-size complexity of auctions, 2012.

\bibitem{HR09}
Jason~D. Hartline and Tim Roughgarden.
\newblock Simple versus optimal mechanisms.
\newblock In {\em Proceedings of the 10th ACM conference on Electronic
  commerce}, EC '09, pages 225--234, New York, NY, USA, 2009. ACM.

\bibitem{HJW}
Darrell Hoy, Kamal Jain, and Christopher~A. Wilkens.
\newblock Coopetitive ad auctions.
\newblock {\em arXiv}, abs/1209.0832, 2012.

\bibitem{K02}
Paul Klemperer.
\newblock What really matters in auction design.
\newblock {\em The Journal of Economic Perspectives}, 16(1):pp. 169--189, 2002.

\bibitem{LST12}
Renato~Paes Leme, Vasilis Syrgkanis, and \'{E}va Tardos.
\newblock Sequential auctions and externalities.
\newblock In {\em Proceedings of the Twenty-Third Annual ACM-SIAM Symposium on
  Discrete Algorithms}, SODA '12, pages 869--886. SIAM, 2012.

\bibitem{L00}
David Lucking-Reiley.
\newblock Vickrey auctions in practice: From nineteenth-century philately to
  twenty-first-century e-commerce.
\newblock {\em Journal of Economic Perspectives}, 14(3):183--192, September
  2000.

\bibitem{MR00}
Eric Maskin and John Riley.
\newblock Asymmetric auctions.
\newblock {\em The Review of Economic Studies}, 67(3):pp. 413--438, 2000.

\bibitem{M04}
P.~Milgrom.
\newblock {\em Putting Auction Theory to Work}.
\newblock Churchill Lectures in Economics. Cambridge University Press, 2004.

\bibitem{M00}
Paul Milgrom.
\newblock Putting auction theory to work: The simulteneous ascending auction.
\newblock {\em Journal of Political Economy}, 108(2):pp. 245--272, 2000.

\bibitem{MR91}
Paul Milgrom and John Roberts.
\newblock Adaptive and sophisticated learning in normal form games.
\newblock {\em Games and Economic Behavior}, 3(1):82 -- 100, 1991.

\bibitem{Rob79}
Kevin Roberts.
\newblock The characterization of implementable social choice rules.
\newblock In \emph{Aggretaion and Revelation of Preferences}, J-J.Laffont
  (ed.), North Holland Publishing Company., 1979.

\bibitem{RT12}
Tim Roughgarden and \'{E}va Tardos.
\newblock Do externalities degrade gsp's efficiency?
\newblock In {\em The Eighth Ad Auctions Workshop}, 2012.

\bibitem{S41}
Paul~A. Samuelson.
\newblock The stability of equilibrium: Comparative statics and dynamics.
\newblock {\em Econometrica}, 9(2):97--120, 1941.

\bibitem{V61}
William Vickrey.
\newblock Counterspeculation, auctions, and competitive sealed tenders.
\newblock {\em The Journal of Finance}, 16(1):8--37, 1961.

\bibitem{W54}
Leon Walras.
\newblock {\em Elements of pure economics, or, The theory of social wealth /
  Leon Walras ; translated by William Jaffe}.
\newblock Published for the American Economic Association and the Royal
  Economic Society by Allen and Unwin, London :, 1954.

\bibitem{WS12}
Christopher~A. Wilkens and Balasubramanian Sivan.
\newblock Single-call mechanisms.
\newblock In {\em Proceedings of the 13th ACM Conference on Electronic
  Commerce}, EC '12, pages 946--963, New York, NY, USA, 2012. ACM.

\end{thebibliography}

\appendix
\section{Selected Proofs}
\label{sec:sel-proofs}

This appendix contains selected proofs that were previously omitted.


\subsection{CEF Claims}

\begin{proofof}{Claim~\ref{clm:welfare-max}} We want to show that $\sum_{i\in\bidders}v_i(o^*)\geq\sum_{i\in\bidders}v_i(o)$ for any outcome $o$ and CEF equilibrium $o^*$. The envy-freeness constraints give
\begin{align*}
\sum_{i\in\bidders}\max\left((v_i(o)-b_i(o))-(v_i(o^*)-b_i(o^*)),0\right)&\leq\sum_{i\in\bidders}b_i(o^*)-b_i(o)\\
\sum_{i\in\bidders}(v_i(o)-b_i(o))-(v_i(o^*)-b_i(o^*))&\leq\sum_{i\in\bidders}b_i(o^*)-b_i(o)\\
\sum_{i\in\bidders}v_i(o)-v_i(o^*)&\leq0
\end{align*}
as desired.
\end{proofof}

\begin{proofof}{Claim~\ref{clm:vcg-dominance}}
Define $o^i$ as the outcome that maximizes the welfare of bidders except $i$:
\[o^i=\argmax_o\sum_{j\neq i}v_j(o)\enspace.\]
Thus, the VCG price of player $i$ is $\sum_{j\neq i}v_j(o^i)-v_j(o^*)$.

Now, the envy-freeness constraints give
\begin{align*}
b_i(o^*)\geq&{} \sum_j\max\left((v_j(o^i)-b_j(o^i))-(v_j(o^*)-b_j(o^*)),0\right)\\
&{}+b_i(o^*)-\sum_{j}(b_j(o^*)-b_j(o^i))\\
b_i(o^*)\geq&{} \sum_{j\neq i}(v_j(o^i)-b_j(o^i))-(v_j(o^*)-b_j(o^*))\\
&{}-\sum_{j\neq i}(b_j(o^*)-b_j(o^i))+b_i(o^i)\\
&{}+\max\left((v_i(o^i)-b_i(o^i))-(v_i(o^*)-b_i(o^*)),0\right)\\
b_i(o^*)&{} \geq\sum_{j\neq i}(v_j(o^i)-v_j(o^*))+b_i(o^i)\\
&{}+\max\left((v_i(o^i)-b_i(o^i))-(v_i(o^*)-b_i(o^*)),0\right)\\
b_i(o^*)&{} \geq\sum_{j\neq i}(v_j(o^i)-v_j(o^*))\enspace.
\end{align*}
\end{proofof}

\begin{proofof}{Claim~\ref{clm:second-price-bound}}
We want to show that
\[\sum_{i\in\bidders}b_i(o^*)\geq\max_o\sum_{i\in\bidders}\max(v_i(o)-v_i(o^*),0)\enspace.\]
For any outcome $o$, the envy-freeness constraints give
\begin{align*}
\sum_{i\in\bidders}\max\left((v_i(o)-b_i(o))-(v_i(o^*)-b_i(o^*)),0\right)&\leq\sum_{i\in\bidders}b_i(o^*)-b_i(o)\\
\sum_{i\in\bidders}\max\left(v_i(o)-v_i(o^*)+b_i(o^*),b_i(o)\right)&\leq\sum_{i\in\bidders}b_i(o^*)\\
\sum_{i\in\bidders}\max\left(v_i(o)-v_i(o^*),0\right)&\leq\sum_{i\in\bidders}\sum_{i\in\bidders}b_i(o^*)
\end{align*}
as desired.
\end{proofof}

\subsection{Convergence Lemmas}

\begin{proofof}{Lemma~\ref{lem:cefconvex}}
Let $b$ be the bids at $\utiltarget$ and $b^\delta$ be the bids at $\utiltarget-\delta$. Note that Claim~\ref{clm:welfare-max} implies a welfare-optimal outcome $o^*$ is winning at $\utiltarget$.

First, suppose that all bidders for whom $\delta_i>0$ are winners at $\utiltarget$. In this case, $v_i(o^*)\geq \utiltarget_i$ and so $v_i(o^*)\geq \utiltarget_i-\delta_i$ and for any outcome $o$ we get
\begin{align*}
\sum_{i\in\bidders}b_i^\delta(o^*)&=\sum_{i\in\bidders}b_i(o^*)+\delta_i\\
&\geq\sum_{i\in\bidders}b_i(o)+\delta_i\\
&\geq\sum_{i\in\bidders}b_i^\delta(o)\enspace,
\end{align*}
implying $o^*$ is still winning at $b_i^\delta$. Since $v_i(o^*)\geq \utiltarget_i-\delta_i$, we can conclude that all bidders are winners, ergo $\utiltarget-\delta\in\cefset$ by Lemma~\ref{lem:allwin-cef}.

Now, suppose some bidders in $\utiltarget$ may be losers, but that the vector $\delta$ has the following property:
\[\delta_i\leq\max(\utiltarget_i-v_i(o*),0)\enspace.\]
This condition says that only losers will raise their bids, and they will not raise them enough to affect $b_i(o^*)$.

Our goal is to show
\[\sum_{i\in\bidders}b_i^\delta(o^*)-b_i^\delta(o)\geq\sum_{i\in\bidders}\max\left((v_i(o)-b_i^\delta(o))-(v_i(o^*)-b_i^\delta(o^*)),0\right)\enspace.\]
First, we see that $b_i^\delta(o)=b_i(o)$ as long as $v_i(o)\leq v_i(o^*)$. For any bidder $i$ we have
\[b_i^\delta(o)=\max(v_i(o)-\utiltarget_i-\delta_i,0)\]
which can only be nonzero if $v_i(o)>\utiltarget_i$. However, $b_i(o)$ can only change if $\delta_i>0$, which requires $\utiltarget_i>v_i(o^*)$ and thus $v_i(o)>\utiltarget_i>v_i(o^*)$. By construction, this also holds for $\utiltarget_i-\delta_i$:
\[v_i(o)>\utiltarget_i-\delta_i\geq v_i(o^*)\enspace.\]

Now, when $v_i(o)>\utiltarget_i-\delta_i\geq v_i(o^*)$, we have
\begin{align*}
(v_i(o)-b_i^\delta(o))-(v_i(o^*)-b_i^\delta(o^*))&=\min(v_i(o),\utiltarget_i-\delta_i)-\min(v_i(o^*),\utiltarget_i-\delta_i)\\
&=\utiltarget_i-\delta_i-v_i(o^*)\\
&\geq0\enspace.
\end{align*}
Importantly, if $\Delta(o)$ is the set of bidders for which $b_i^\delta(o)\neq b_i(o)$, we may conclude that
\[\sum_{i\in\bidders}\max\left((v_i(o)-b_i^\delta(o))-(v_i(o^*)-b_i^\delta(o^*)),0\right)=\]
\begin{align*}
=&{}\sum_{i\not\in\Delta(o)}\max\left((v_i(o)-b_i^\delta(o))-(v_i(o^*)-b_i^\delta(o^*)),0\right)\\
&{}+\sum_{i\in\Delta(o)}(v_i(o)-b_i^\delta(o))-(v_i(o^*)-b_i^\delta(o^*))\\
\end{align*}
and likewise
\[\sum_{i\in\bidders}\max\left((v_i(o)-b_i(o))-(v_i(o^*)-b_i(o^*)),0\right)=\]
\begin{align*}
=&{}\sum_{i\not\in\Delta(o)}\max\left((v_i(o)-b_i(o))-(v_i(o^*)-b_i(o^*)),0\right)\\
&{}+\sum_{i\in\Delta(o)}(v_i(o)-b_i(o))-(v_i(o^*)-b_i(o^*))\enspace.\\
\end{align*}
The desired CEF condition quickly follows, using the fact that bidders $i\not\in\Delta(o)$ did not change their bids:
\[\sum_{i\in\bidders}\max\left((v_i(o)-b_i^\delta(o))-(v_i(o^*)-b_i^\delta(o^*)),0\right)=\]
\begin{align*}
=&{}\sum_{i\not\in\Delta(o)}\max\left((v_i(o)-b_i^\delta(o))-(v_i(o^*)-b_i^\delta(o^*)),0\right)+\sum_{i\in\Delta(o)}(v_i(o)-b_i^\delta(o))-(v_i(o^*)-b_i^\delta(o^*))\\
=&{}\sum_{i\not\in\Delta(o)}\max\left((v_i(o)-b_i(o))-(v_i(o^*)-b_i(o^*)),0\right)+\sum_{i\in\Delta(o)}(v_i(o)-b_i(o))-(v_i(o^*)-b_i(o^*))\\
&{}+\sum_{i\in\Delta(o)}(b_i(o)-b_i^\delta(o))-(b_i(o^*)-b_i^\delta(o^*))\\
=&{}\sum_{i\in\bidders}\max\left((v_i(o)-b_i(o))-(v_i(o^*)-b_i(o^*)),0\right)+\sum_{i\in\bidders}(b_i(o)-b_i^\delta(o))-(b_i(o^*)-b_i^\delta(o^*))\\
\leq&{}\sum_{i\in\bidders}(b_i(o^*)-b_i(o)+\sum_{i\in\bidders}(b_i(o)-b_i^\delta(o))-(b_i(o^*)-b_i^\delta(o^*))\\
=&{}\sum_{i\in\bidders}b_i^\delta(o^*)-b_i^\delta(o)
\end{align*}
as desired.

Finally, for general $\delta$, split it as $\delta=\delta^1+\delta^2$ where
\[\delta^1_i=\min(\delta_i,\utiltarget_i-v_i(o^*))\enspace.\]
The vector $\delta^1$ satisfies the condition $\delta_i\leq\max(\utiltarget_i-v_i(o*),0)$, so $\utiltarget-\delta^1\in\cefset$. Moreover, all bidders are winners in $\utiltarget-\delta^1$, so
\[s-\delta^1-\delta^2=s-\delta\in\cefset\]
as desired.
\end{proofof}

\begin{proofof}{Lemma~\ref{lem:allwin-finite}}
Properties \Aa{} and \Ac{} imply that a bid will only be lowered if there are no losers. Thus, bids will only be raised (utility-targets decreased) until all bidders are simultaneously winners. Since any bidder $i$ is always a winner when bidding $\utiltarget_i=0$ and bidders never decrease their utility-targets when they are winners (\Aa), utility-target can be decreased at most $\left|\left\lceil\frac1\epsilon s\right\rceil\right|_1$ times before all bidders are winners. Moreover, since losers will always try to decrease their utility-targets (\Aa), the auction will never stall in a configuration where some bidder is a loser.
\end{proofof}

\begin{proofof}{Lemma~\ref{lem:notallwin-finite}} If the auction reaches a vector $\utiltarget$ that induces an outcome $o\neq o^*$, then $\utiltarget\in\ncefset$ and we are done. Thus, it remains to show that an auction will reach a vector $\utiltarget\in\ncefset$ even if the outcome is always $o^*$.

Consider a bidder $j$. By \Aa{} we know that $j$ will only decrease $\utiltarget_j$ if she is a loser and increase $\utiltarget_j$ if she is a winner. By \Ac{} we can conclude that $j$ will eventually decrease her bid until $\utiltarget_j\geq v_j(o^*)-\epsilon$, implying $b_j(o^*)\leq\epsilon$. Thus,
\[\sum_{i\in\bidders}b_i(o^*)-b_i(o)\leq n\epsilon\enspace.\]
Now, as long as there is some outcome $o$ and bidder $j$ such that $v_j(o)>v_j(o^*)$, when $\epsilon$ is sufficiently small it will be the case that 
\[n\epsilon<\sum_{i\in\bidders}\max\left((v_i(o)-b_i(o))-(v_i(o^*)-b_i(o^*)),0\right)\]
Unfortunately, this implies
\[\sum_{i\in\bidders}\max\left((v_i(o)-b_i(o))-(v_i(o^*)-b_i(o^*)),0\right)>\sum_{i\in\bidders}b_i(o^*)-b_i(o)\enspace,\]
and therefore $\utiltarget\in\ncefset$.

If there is no outcome $o$ and bidder $j$ such that $v_j(o)>v_j(o^*)$, then by similar logic bidders will increase their utility-targets until precisely $\utiltarget_j=v_j(o^*)$ (as a result of our restriction that bidders always bid $\utiltarget_j\leq\sup_ov_j(o^*)$).
\end{proofof}

\noindent{\bf Theorem \ref{thm:converge-out-sink} (Restatement).}
{\em 
If winners try to lower their effective bids (\Ac) and losers try raising but not lowering (\Aa), bids will eventually remain close to the boundary of the set of CEF bids or entirely outside it.
}\\

\begin{proofof}{Theorem~\ref{thm:converge-out-sink}}
By Lemma~\ref{lem:notallwin-finite}, the auction will eventually reach a utility-target vector in $\ncefset$ or the degenerate case where nobody is paying anything and $o^*$ is winning. In the degenerate case, bids converge to a point on the boundary of $\cefset$, so the theorem is true. For the standard case, we show that $\utiltarget\in\nearncef$ from the first time a bid in $\ncefset$ is reached.

If $\utiltarget\in\ncefset$, we are done, so suppose $\utiltarget\in\cefset$. Let $\utiltarget'$ be the most recent utility-targets that were in $\ncefset$ and let $\utiltarget''\in\cefset$ be the utility-targets immediately after $\utiltarget'$. Let $i$ be the bidder who changed her bid between $\utiltarget'$ and $\utiltarget''$. Corollary~\ref{cor:ncefconvex} implies that $i$ must have raised her bid between $\utiltarget'$ and $\utiltarget''$.

First, suppose that the outcome changed from $o'$ to $o''$ when $i$ raised her bid. Since $o^*$ must be the outcome of any CEF bid, we know that $o''=o^*$ and that the outcome does not change again before bids reach $\utiltarget$. Define the utility-target vector $\tilde s$ with associated bids $\tilde b$ as follows:
\[\tilde \utiltarget_j=\begin{cases}\min(\utiltarget_i,\utiltarget_i')&j=i\\
\utiltarget_j-\epsilon&\utiltarget_j>\utiltarget_j''\\
\utiltarget_j+\epsilon&\utiltarget_j<\utiltarget_j'\\
\utiltarget_j&\mbox{otherwise.}\end{cases}\]
Let $\delta_j=\tilde \utiltarget_j-\utiltarget_j''$. We argue later that $i$ will not increase her utility-target from $\utiltarget_i''=\utiltarget_i'+\epsilon$, so $|\tilde \utiltarget_j-\utiltarget_j|\leq\epsilon$ for all $j$. Thus, it is sufficient to show that $\tilde s\in\ncefset$.

Consider a bidder $j\neq i$ and suppose $\utiltarget_j>\utiltarget_j''$. By definition, we get a simple bound on $j$'s bid for $o'$:
\[\tilde b_j(o')\geq b_j'(o')-\delta_j\enspace.\]
We also know that $j$ lowered her bid at some point between $\utiltarget''$ and $\utiltarget$. Since $j$ would only increase her utility-target if she were a winner, she must have been a winner at some value $\geq \utiltarget_j-\epsilon=\tilde \utiltarget_j$. Thus, $\tilde \utiltarget_j=\utiltarget_j-\epsilon\leq v_j(o^*)$. We can thus upper-bound her bid for $o^*$:
\[\tilde b_j(o^*)\leq b_j'(o^*)-\delta_j\enspace.\]
Combining these two bounds and noting that $b_j''=b_j'$ for $j\neq i$ gives
\[\tilde b_j(o')-\tilde b_j(o^*)\geq b_j'(o')-b_j'(o^*)\enspace.\]

For bidders $j\neq i$ with $\utiltarget_j<\utiltarget_j''$, analogous reasoning based the fact that $j$ must have been a loser to decrease her utility-target gives
\[\tilde b_j(o')-\tilde b_j(o^*)\geq b_j'(o')-b_j'(o^*)\enspace.\]
For bidders $j\neq i$ with $\utiltarget_j=\utiltarget_j''$, we trivially have
\[\tilde b_j(o')-\tilde b_j(o^*)= b_j'(o')-b_j'(o^*)\enspace,\]
so it remains to consider bidder $i$.

For bidder $i$, we know that decreasing her utility-target from $\utiltarget_i'$ to $\utiltarget_i''$ increased her bid for $o^*$ more than it increased her bid for $o'$. This implies $v_i(o')<v_i(o^*)$ and $\utiltarget_i''<v_i(o*)$. Consequently, $i$ is a winner with $\utiltarget_i''$ at $o^*$ and will not decrease her utility-target further. Firstly, this implies that $|\tilde \utiltarget_i-\utiltarget_i|\leq\epsilon$. First, suppose $\utiltarget_i>\utiltarget_i''$. In this case, $\utiltarget_i\geq \utiltarget_i'$, and since $v_i(o^*)>v_i(o')$ we have
\[\tilde b_i(o')-\tilde b_i(o^*)\geq b_i'(o')-b_i'(o^*)\enspace.\]
Otherwise, $i$ does not change her bid from $\utiltarget''$ to $\utiltarget$, so $\tilde \utiltarget_i=\utiltarget_i'$ and therefore $\tilde b_i=b_i'$

Thus, for any bidder $j$ we have
\[\tilde b_j(o')-\tilde b_j(o^*)\geq b_j'(o')-b_j'(o^*)\enspace,\]
and thus
\[\sum_{j\in\bidders}\tilde b_j(o')-\sum_{j\in\bidders}\tilde b_j(o^*)\geq \sum_{j\in\bidders}b_j'(o')-\sum_{j\in\bidders}b_j'(o^*)\enspace.\]
Since $o'$ was winning at $b_j'$, this implies $o^*$ cannot be winning under $\tilde \utiltarget$, and therefore $\tilde \utiltarget\in\ncefset$. By construction, $|\tilde \utiltarget_j-\utiltarget_j|\leq\epsilon$, so this implies $\utiltarget\in\nearncef$.

So far, we showed that $\utiltarget\in\nearncef$ as long as the outcome changed when $i$ raised her bid. In the case where the outcome was already $o'=o^*$, we want to analyze the CEF constraints directly. Since $\utiltarget'\in\ncefset$, there is some outcome $o^w$ for which the CEF constraints are violated, i.e.
\[\sum_{i\in\bidders}b_i'(o^*)-b_i'(o^w)<\sum_{i\in\bidders}\max\left((v_i(o^w)-b_i'(o^w))-(v_i(o^*)-b_i'(o^*)),0\right)\enspace.\]
Observing that the outcome does not change from $\utiltarget'$ to $\utiltarget$, the logic from the case where $o'\neq o^*$ gives
\[\tilde b_j(o')-\tilde b_j(o^*)\geq b_j'(o')-b_j'(o^*)\]
for any bidder $j$. It immediately follows that
\[\sum_{i\in\bidders}\tilde b_i(o^*)-\tilde b_i(o^w)<\sum_{i\in\bidders}\max\left((v_i(o^w)-\tilde b_i(o^w))-(v_i(o^*)-\tilde b_i(o^*)),0\right)\enspace,\]
and so $\tilde \utiltarget\in\ncefset$ and $\utiltarget\in\nearncef$.
\end{proofof}

\subsection{Convergence to the Egalitarian Equilibrium}


\noindent{\bf Theorem \ref{thm:converge-egal} (Restatement).}
{\em 
If losing bidders will raise their effective bids (\Aa), winning bidders will try lowering their effective bids (\Ac), and the most impatient bidder is the losing bidder bidding for the highest utility (\Ab, \Ad), then bids will converge to the Egalitarian envy-free equilibrium.
}\\

\begin{proofof}{Theorem~\ref{thm:converge-egal}}
The proof will proceed as follows. We first categorize bidders into levels based on their utility in the egalitarian outcome. We define upper and lower bounds on utility-targets as multiples of $\epsilon$, the amount by which players change their bids. Then, we show that if for a given bidder $j$, the bid of every lower-utility bidder has converged to within their bounds, the bid of $j$ will also converge to within her bounds - first to at least her lower bound (Lemma~\ref{lem:sbid-lowerbound}), and then to at most her upper bound (Lemma~\ref{lem:sbid-upperbound}). Combining these via induction gives our final result that the bids of all players converge near their egalitarian outcome. 

Let $\optoutcome$ be the egalitarian outcome; let $\optsbid{i}$ be the corresponding utility-target of bidder $i$. Let $\mathbf{B}_{b}(\outcome)=\sum_{i\in\bidders} b_i(\outcome)$ be the total bid for a given outcome $\outcome$. Let $\mathbf{B}_{X}(\outcome) = \sum_{j\in \level_{X}} b_j(\outcome)$, and $\mathbf{B}^*_{X}(\outcome)$ be similarly defined.

First, consider all utility-targets in the egalitarian equlibrium; let $z_i$ be the $i$th smallest (distinct) utility-target. Let $\level_i$ be the set of all players with a utility-target of $z_i$ in the egalitarian equilibrium. We will use $\level(j)$ to denote the level of a bidder $j$.

We will show convergence by showing that there exist functions $\lowbound(\levelidx)$ and $\highbound(\levelidx)$ s.t. for any $j\in \level_{\levelidx}$, utility-targets converge into and remain in the interval $[\utiltarget_j^* - \epsilon \lowbound(\levelidx), \utiltarget_j^* + \epsilon \highbound(\levelidx)]$.

\textbf{Bidding Bounds.}
We now precisely define the bounding functions $\lowbound(\cdot)$ and $\highbound(\cdot)$.

\begin{definition}
\label{def:bounds}
\begin{align}
\lowbound(i)& =  2^{2|\level_{<i}| }\\
\highbound(i)& =  2^{2|\level_{<i}| + |\level_i|}
\end{align}
\end{definition}

These bounds are given specifically so that for any level $k$, the sum over upper bounds in lower levels is at most the lower bound in level $k$, and the sum over all lower bounds for lower (or equal) levels is at most the upper bound for level $k$. Intuitively, we are saying that lower-level bidders cannot over or under bid enough to make up for bidders in level $k$.
 
\begin{claim}
\label{clm:bbounds}
\begin{align}
\lowbound(k) & >  \sum_{i=0}^{k-1} |L_i| \highbound(i) 	\label{eq:bminusbound}\\
\highbound(k) & >  \sum_{i=0}^{k} |L_i| \lowbound(i) \label{eq:bplusbound}
\end{align}
\end{claim}

We omit the proof; it follows from manipulation of exponential sums.

\textbf{Witness outcomes.}
Recall from Algorithm~\ref{alg:s-egal} that utility-targets for any given player are raised until the CEF constraint for some outcome $\outcome$ is violated. These outcomes have an important role to play in the egalitarian equilibrium --- they are the reason that a bidder cannot achieve any more utility. We will call them \emph{witness} outcomes.

Three properties of these witness outcomes are important for us. First, bidder $i$ values the witness at less than her egalitarian bid, hence she would be `losing' if it was chosen above the egalitarian winning ad; that all bidders with higher utility value it at at least their utility-target; and that with the final winning bids, the total bid of each is tied. We define \emph{witness} outcomes precisely as follows:

\begin{definition}
\label{def:witness}
Outcome $\witness$ is a \emph{witness} outcome for bidder $\bidderidx$ at the egalitarian utility-targets $\optsbid{}$ if its total bid is tied with the egalitarian outcome, $\bidderidx$ asking for more utility at the egalitarian equilibrium results in a higher total bid for $\witness$ than for the optimal egalitarian ad and $\bidderidx$ is the highest-utility bidder to lose if $\witness$ wins over $\optoutcome$. 
\end{definition}

Recall the intuition behind these outcomes: they are the reason that a bidder cannot achieve more utility at the egalitarian equilibrium.  If there is no witness for a bidder who must pay something, then the bidder could ask for more utility, and higher utility bidders could effectively `pick up the slack', resulting in a more egalitarian outcome.

\begin{claim}
\label{clm:at-least-one-witness}
At the egalitarian equilibrium, every bidder $\bidderidx$ s.t. $\optsbid{i} < \val{i}(\optoutcome)$ has at least one witness.
\end{claim}

\begin{proof}
We will prove via contradiction. Assume at the CEF egalitarian outcome $\optoutcome$ bidder $\bidderidx$ has no witness. Now, let bidder $\bidderidx$ increase her utility-target by a small enough $\epsilon>0$, that only outcomes that were previously tied with $\optoutcome$ win over $\optoutcome$. For each of these outcomes, there must be a higher utility bidder than $\bidderidx$ who does not win with the outcome; otherwise it would be a witness for $\bidderidx$. Decrease the utility-targets of the highest utility bidder not in each of these outcomes by $\epsilon$. At this point, all outcomes will be tied again --- and we can have the optimal outcome win the tiebreaker via having a higher utility, or assume that one player will decrease, then raise their utility-target to ensure that it was the previous outcome to win. These bids will be CEF, and will be more egalitarian than $\optoutcome$, as bidder $\bidderidx$ achieved more utility, and only higher utility bidders achieved less utility.
\end{proof}

Another important property will be that each outcome is only a witness for bidders of a single level:

\begin{claim}
\label{clm:one-level-per-witness}
An outcome is only a witness outcome for bidders of a single level.
\end{claim}

This really follows from the definition --- bidders in different levels cannot both be the highest utility bidder to not win with an outcome. More intuitively though, if players of different levels were both not in an outcome, and the lower utility bidder had no other witness outcome, then a more egalitarian outcome would involve increasing his utility-target, and decreasing the utility-target of the higher utility bidder.

\textbf{Bidding convergence.}
We now present the core of our convergence result. This convergence is a two step process for bidders in a given level; after the utility bids of all lower level bidders have converged within their bounds, convergence in the given level to at least the lower bound takes place first, and then bids in the given level will converge to below their upper bound.
	
\begin{lemma}
\label{lem:sbid-lowerbound}
Under assumptions \Aa, \Ab, \Ac\  and \Ad, the utility-target of each bidder $\bidderidx$ in level $\level_{\levelidx}$ will converge to at least their lower bounds, $\optsbid{\bidderidx} - \epsilon \cdot \lowbound(\levelidx)$ if for every bidder $\altbidderidx$ in level $\level_{\lowlevelidx}$ s.t. $\lowlevelidx < \levelidx$, $\sbid{\altbidderidx} \leq  \optsbid{\altbidderidx} + \epsilon \cdot \highbound(\lowlevelidx)$.
\end{lemma}

\begin{proof}
Our argument consists of two parts: first, that if a bidder is bidding for utility at or below her lower bound then she will never reduce her utility-target further. Second, she will eventually try raising her bid (by \Ac). These two combined will lead to her eventually raising her bid to at least the lower bound.

\begin{claim}
Under the assumptions of Lemma \ref{lem:sbid-lowerbound}, no bidder $\bidderidx$ in level $\level_{\levelidx}$ with a utility-target of $\sbid{\bidderidx} \leq \optsbid{\bidderidx} - \epsilon \cdot \lowbound(\levelidx)$ will lower her utility-target.
\end{claim}

We will prove via contradiction. Assume for bidder $\bidderidx$ that with a utility-target of $\sbid{\bidderidx} \leq \optsbid{\bidderidx} - \epsilon \cdot \lowbound(\levelidx)$, she wishes to lower her utility-target further. Let $\outcome$ be the winning outcome with bids $\allsbid$. As $i$ will only lower her bid if she is losing (\Aa), $\sbid{\bidderidx} > \val{\bidderidx}(\outcome)$. We will now try to derive the contradiction that the total bid for the optimal outcome is at least the total effective bid for $\outcome$ ($\totalbid(\optoutcome) > \totalbid(\outcome)$), hence she must win and would not care to lower her utility-target.

For $\bidderidx$ to decrease her utility-target, by \Ad\ she must be the highest utility bidder who is losing such that $\sbid{\bidderidx} > \val{\bidderidx}(\outcome)$. By our bound, we know that for every lower utility bidder $\altbidderidx$ in level $\level_{\lowlevelidx}$, $\sbid{\altbidderidx}\leq \optsbid{\altbidderidx} + \epsilon \highbound(\lowlevelidx)$. Since $\optoutcome$ is the optimal winning outcome and $\outcome$ the currently winning outcome, $\mathbf{B}^*(\optoutcome) \geq \mathbf{B}^*(\outcome)$ and $\mathbf{B}(\optoutcome) \leq \mathbf{B}(\outcome)$.

In the egalitarian outcome, every bidder $\bidderidx$ receives the utility she bids for; hence $b_{\bidderidx}^*(\optoutcome) = \val{\bidderidx}(\optoutcome) - \optsbid{\bidderidx}$. By our assumption on the utility-target bounds, for all bidders $\altbidderidx\in \level_{<\levelidx+1}$, $b^*_{\altbidderidx}(\optoutcome) - b_{\altbidderidx}(\optoutcome) \leq \epsilon \highbound(\level({\altbidderidx}))$.

Consider a bidder $\altbidderidx$ in a lower level than $\bidderidx$ and first, is requesting more utility relative to the egalitarian outcome, specifically that $\sbid{\altbidderidx} > \optsbid{\altbidderidx}$. Hence, we will have $0\geq b_{\altbidderidx}(\optoutcome) - b_{\altbidderidx}^*(\optoutcome) \geq -(\sbid{\altbidderidx}-\optsbid{\altbidderidx})$ and $0\geq b_{\altbidderidx}(\outcome) - b_{\altbidderidx}^*(\outcome) \geq -(\sbid{\altbidderidx}-\optsbid{\altbidderidx})$. Hence,
\begin{align} 
(b_{\altbidderidx}(\optoutcome) - b_{\altbidderidx}^*(\optoutcome))-(b_{\altbidderidx}(\outcome) - b_{\altbidderidx}^*(\outcome)) &\geq -(\sbid{\altbidderidx}-\optsbid{\altbidderidx})\\
 					&\geq -\highbound(\level(\altbidderidx)). \label{eq:lowerbound-surplusbid}
\end{align}

Consider the case that $\sbid{\altbidderidx} \leq \optsbid{\altbidderidx}$, that $\altbidderidx$ is requesting less utility than in the egalitarian outcome. Then $b_{\altbidderidx}(\optoutcome) - b_{\altbidderidx}^*(\optoutcome) = -(\sbid{\altbidderidx} - \optsbid{\altbidderidx})\geq 0 $, and $b_{\altbidderidx}(\outcome) - b_{\altbidderidx}^*(\outcome)\leq -(\sbid{\altbidderidx} - \optsbid{\altbidderidx})$. Hence, 

\begin{align}
(b_{\altbidderidx}(\optoutcome) - b_{\altbidderidx}^*(\optoutcome))-(b_{\altbidderidx}(\outcome) - b_{\altbidderidx}^*(\outcome)) \geq 0. \label{eq:lowerbound-lesssurplusbid}
\end{align}

Summing over all lower-level bidders via Equations \eqref{eq:lowerbound-surplusbid} and \eqref{eq:lowerbound-lesssurplusbid} gives $(\totalbid_{<\levelidx}(\optoutcome) - \totalbid_{<\levelidx}^*(\optoutcome)) - (\totalbid_{<\levelidx}(\outcome) - \totalbid_{<\levelidx}^*(\outcome)) \geq -\sum_{\altlevelidx < \level(\bidderidx)} \highbound(\altlevelidx)$ and hence by Claim~\ref{clm:bbounds},
\begin{equation}\label{eq:lowerbidderbound}
(\totalbid_{<\levelidx}(\optoutcome) - \totalbid_{<\levelidx}^*(\optoutcome)) - (\totalbid_{<\levelidx}(\outcome) - \totalbid_{<\levelidx}^*(\outcome)) > - \lowbound(\levelidx).
\end{equation}

Now, consider a bidder $\altbidderidx$ in the same or a higher level than $\bidderidx$. If $\altbidderidx$ is overbidding and not winning in outcome $\outcome$ with bids $b$, then she would have decreased her utility-target faster than $\bidderidx$. She could however be overbidding and winning in $\outcome$; in which case the decrease in bids for $\optoutcome$ must be bounded by the decrease for $\outcome$, hence: $(b_{\altbidderidx}(\optoutcome) - b_{\altbidderidx}^*(\optoutcome)) - (b_{\altbidderidx}(\outcome) - b_{\altbidderidx}^*(\outcome)) \geq 0$. If she is requesting less utility, $\optoutcome$ will see the full increase in bid while $\outcome$ may not. Denote the total bid of all bidders aside from $\bidderidx$ in the same or higher level as $\bidderidx$ as $\totalbid_{\geq \levelidx\setminus \bidderidx}(\outcome)$. Then, summing over all such bidders gives

\begin{equation}\label{eq:higherbidderbound}
\totalbid_{\geq \levelidx\setminus \bidderidx}(\optoutcome) - \totalbid_{\geq \levelidx\setminus \bidderidx}^*(\optoutcome)) - (\totalbid_{\geq \levelidx\setminus \bidderidx}(\outcome) - \totalbid_{\geq \levelidx\setminus \bidderidx}^*(\outcome)) \geq 0.
\end{equation}

Our original assumption on $\bidderidx$ gives $(b_{\bidderidx}(\optoutcome) - b_{\bidderidx}^*(\optoutcome)) - (b_{\bidderidx}(\outcome) - b_{\bidderidx}^*(\outcome))\leq \lowbound(\levelidx)$. Now, taking the sum over this and equations \eqref{eq:lowerbidderbound} and \eqref{eq:higherbidderbound} gives $(\totalbid(\optoutcome) - \totalbid^*(\optoutcome)) - (\totalbid(\outcome) - \totalbid^*(\outcome)) > -\highbound(\level(j)) + \highbound(\level(j))=0$. By our assumption that $\optoutcome$ is the egalitarian winning outcome, we have $\totalbid^*(\optoutcome) - \totalbid^*(\outcome) \geq 0$. Adding these yields
\begin{equation}
\totalbid(\optoutcome) - \totalbid(\outcome) > 0.
\end{equation}
This is in violation of our assumption that $\outcome$ wins with bids $b$. Hence, no such bidder $\bidderidx$ can ever wish to lower her utility-target past the lower bound when all lower-level agents have bids within their upper bounds. By Assumption \Ac,  she will eventually try and lower her bid when winning, hence her bid will converge above her lower bound.
\end{proof}	

\begin{lemma}
\label{lem:sbid-upperbound}
Under assumptions \Aa, \Ab, \Ac\ and \Ad, the utility-target $\sbid{\bidderidx}$ of each bidder $\bidderidx$ in level $\level_{\levelidx}$ will converge to at most the upper bound, $\optsbid{\bidderidx} + \epsilon \cdot \highbound(\levelidx)$ if for every bidder $\altbidderidx$ in level $\level_{\lowlevelidx}$ s.t. $\lowlevelidx\leq \levelidx$, $\utiltarget_{\altbidderidx} \geq  s^*_{\altbidderidx} - \epsilon \cdot \lowbound(\lowlevelidx)$. 
\end{lemma}		

\begin{proof}

By Assumptions \Aa, \Ab\ and Observation \ref{obs:lower-allwin}, a bidder will only request more utility from a set of bids $b$ with winning outcome $\outcome$ if all other bidders are winning with bids $b$, and by Lemma \ref{lem:allwin-cef}, $b$ must be CEF.%

Our proof will proceed by showing that in any such $\outcome$, $\sbid{\bidderidx} < \optsbid{\bidderidx} + \epsilon \cdot \highbound(\levelidx)$, and hence her utility-target must stay below $\optsbid{\bidderidx} + \epsilon \cdot \highbound(\levelidx)$ in winning outcomes. Furthermore, by Theorem \ref{thm:converge-sink} bids will become CEF; hence $i$ will be forced to decrease her utility-target.

By Claim~\ref{clm:at-least-one-witness}, there is a witness outcome $\witness$ which includes every bidder $\altbidderidx$ in a strictly higher level $\highlevelidx$ than $\bidderidx$.  We will now show that if all other players are winning with the egalitarian winning outcome, then $\bidderidx$'s utility-target must be below her upper bound, otherwise the witness outcome $\witness$ would win over $\optoutcome$.

By Definition \ref{def:witness}, $\totalbid^*(\witness) = \totalbid^*(\optoutcome)$. Consider the quantity $\totalbid(\optoutcome) - \totalbid^*(\optoutcome)$, and break it into sums over bidders in levels at or below bidder $\bidderidx$, $\bidderidx$ and bidders in levels above $\bidderidx$:
\begin{align*}
(\totalbid(\optoutcome) - \totalbid^*(\optoutcome)) 
	= &(\totalbid_{\leq \levelidx \setminus \bidderidx}(\optoutcome) - \totalbid^*_{\leq \levelidx \setminus \bidderidx}(\optoutcome)) \\
	&+ (b_{\bidderidx}(\optoutcome) - b_{\bidderidx}^*(\optoutcome))\\ 
	&+ (\totalbid_{>\levelidx}(\optoutcome) - \totalbid^*_{>\levelidx}(\optoutcome))
\end{align*}

We will now proceed by separately considering bidders in higher and lower levels than bidder $\bidderidx$. We will bound the change in bids from each, and see that there is no way for bidder $\bidderidx$ to ask for utility above her upper bound and still be in the winning outcome.

\textbf{Higher-level bidders.}
By properties of witness sets, any such bidder $\altbidderidx$ must be winning in the witness outcome at both the egalitarian bids and the current bids, hence $b_{\altbidderidx}(\witness) - b_{\altbidderidx}^*(\witness) = -(\sbid{\altbidderidx} - \optsbid{\altbidderidx})$. Since we know that at the egalitarian bids, such a bidder must be winning in the egalitarian outcome,  $b_{\altbidderidx}(\optoutcome) - b_{\altbidderidx}^*(\optoutcome) = -(\sbid{\altbidderidx} - \optsbid{\altbidderidx}) \leq  b_{\altbidderidx}(\witness) - b_{\altbidderidx}^*(\witness)$. Summing over all such bidders yields 
\begin{equation}
(\totalbid_{> \levelidx}(\optoutcome)-\totalbid^*_{> \levelidx}(\optoutcome)) - (\totalbid_{> \levelidx}(\witness)-\totalbid^*_{> \levelidx}(\witness)) \leq 0 \label{eq:highbidderser}
\end{equation}

\textbf{Lower-level bidders.}
By our initial assumption that bidding has converged above lower bounds for these bidders, for any bidder $\altbidderidx$ in $\level_{\leq \levelidx}$, $ \sbid{\altbidderidx} \geq \optsbid{\altbidderidx} - \epsilon \lowbound(\level(\altbidderidx))$, and hence $b_{\altbidderidx}(\optoutcome) \leq b_{\altbidderidx}^*(\optoutcome) + \epsilon \lowbound(\level(\altbidderidx))$ and $b_{\altbidderidx}(\witness) \leq b_{\altbidderidx}^*(\witness) + \epsilon \lowbound(\level(\altbidderidx))$. 

Recall that all bids $b_{\altbidderidx}^*(\optoutcome)$ and $b_{\altbidderidx}(\optoutcome)$ are winning by assumption --- since $\optoutcome$ is the egalitarian outcome, and no player wishes to decrease their utility-target in the current bids. If for some bidder $\altbidderidx$, $v_{\altbidderidx}(\witness) \geq v_{\altbidderidx}(\optoutcome)$, then $(b_{\altbidderidx}(\witness)- b_{\altbidderidx}^*(\witness)) = (b_{\altbidderidx}(\optoutcome)- b_{\altbidderidx}^*(\optoutcome)) = -(\utiltarget_{\altbidderidx} - \utiltarget_{\altbidderidx}^*)$.

Consider then the case that $\val{\altbidderidx}(\witness) < \val{\altbidderidx}(\optoutcome)$; that is, that $\altbidderidx$ values the witness $\outcome$ less than the egalitarian outcome. We will consider two cases: that her utility-target is lower or higher than her egalitarian utility-target respectively. 
\begin{enumerate}

\item[$\mathbf{(\sbid{\altbidderidx} < \optsbid{\altbidderidx})}$] If the bidder $\altbidderidx$ bids for less utility than in the egalitarian outcome, then that increase in effective bid will be bounded by the increase in the bid for the egalitarian outcome. That is, we have $b_{\altbidderidx}(\witness) - b_{\altbidderidx}^*(\witness) = \max(v_{\altbidderidx}(\witness), \utiltarget_{\altbidderidx}) - \utiltarget_{\altbidderidx} - \max(v_{\altbidderidx}(\witness), \utiltarget_{\altbidderidx}^*) + \utiltarget_{\altbidderidx}^*$, and hence $b_{\altbidderidx}(\witness) - b_{\altbidderidx}^*(\witness) = -(\utiltarget_{\altbidderidx}-\utiltarget_{\altbidderidx}^*) + (\max(v_{\altbidderidx}(\witness), \utiltarget_{\altbidderidx}) - \max(v_{\altbidderidx}(\witness), \utiltarget_{\altbidderidx}^*))$. As $b_{\altbidderidx}(\optoutcome) - b_{\altbidderidx}^*(\optoutcome) = -(\utiltarget_{\altbidderidx}-\utiltarget_{\altbidderidx}^*)$, we then have:
\begin{equation}
0 \leq b_{\altbidderidx}(\witness) - b_{\altbidderidx}^*(\witness) \leq b_{\altbidderidx}(\optoutcome) - b_{\altbidderidx}^*(\optoutcome) = -(\utiltarget_{\altbidderidx} - \utiltarget_{\altbidderidx}^*).
\end{equation} 
Furthermore, since $\utiltarget_{\altbidderidx}\geq \utiltarget_{\altbidderidx}^* - \epsilon \lowbound(L({\altbidderidx}))$ by assumption, we have:
\begin{equation}
0 \leq b_{\altbidderidx}(\witness) - b_{\altbidderidx}^*(\witness) \leq b_{\altbidderidx}(\optoutcome) - b_{\altbidderidx}^*(\optoutcome) \leq \epsilon \lowbound(L({\altbidderidx}))
\end{equation}
and
\begin{equation}\label{eq:smallersurplusbid}
0 \leq (b_{\altbidderidx}(\optoutcome) - b_{\altbidderidx}^*(\optoutcome)) - (b_{\altbidderidx}(\witness) - b_{\altbidderidx}^*(\witness)) \leq \epsilon \lowbound(L({\altbidderidx})).
\end{equation}

\item[$\mathbf{(\utiltarget_{\altbidderidx} \geq \utiltarget_{\altbidderidx}^*)}$] If bidder ${\altbidderidx}$ instead is bidding for at least as much utility as in the egalitarian outcome, the decrease in total bid is bounded by the change in bids for the egalitarian outcome, hence the change in utility-targets will be between $b_{\altbidderidx}(\optoutcome) - b_{\altbidderidx}^*(\optoutcome) = \utiltarget_{\altbidderidx}^*-\utiltarget_{\altbidderidx}$ and $0$. Hence, 
\begin{equation}
b_{\altbidderidx}(\optoutcome) - b_{\altbidderidx}^*(\optoutcome) \leq b_{\altbidderidx}(\witness) - b_{\altbidderidx}^*(\witness) \leq 0
\end{equation}
and 
\begin{equation}\label{eq:largersurplusbid}
-(\utiltarget_{\altbidderidx}-\utiltarget_{\altbidderidx}^*) \leq (b_{\altbidderidx}(\optoutcome) - b_{\altbidderidx}^*(\optoutcome)) - (b_{\altbidderidx}(\witness) - b_j^*(\witness)) \leq 0.
\end{equation}
\end{enumerate}

We now have upper bounds on $-(\utiltarget_{\altbidderidx}-\utiltarget_{\altbidderidx}^*) \leq (b_{\altbidderidx}(\optoutcome) - b_{\altbidderidx}^*(\optoutcome)) - (b_{\altbidderidx}(\witness) - b_j^*(\witness))$ for all lower-level bidders. Taking the sum across all members of $\level_{\leq \levelidx}$ via Equations \eqref{eq:smallersurplusbid} and \eqref{eq:largersurplusbid} gives:
\begin{equation}
\sum_{\altbidderidx \in \level_{\leq \levelidx}} (b_{\altbidderidx}(\optoutcome) - b_{\altbidderidx}^*(\optoutcome)) - (b_{\altbidderidx}(\witness) - b_{\altbidderidx}^*(\witness)) \leq \sum_{\altbidderidx\in \level_{\leq \levelidx}} \epsilon \lowbound(L(\altbidderidx))
\end{equation}

Rearranging and noting that $\sum_{\altbidderidx\in \level_{\leq \levelidx}} \epsilon \lowbound(L(\altbidderidx)) < \highbound(\levelidx)$ by Claim \ref{clm:bbounds} gives
\begin{equation}
\label{eq:smallerlevel-increasebound}
(\totalbid_{\leq \levelidx \setminus \bidderidx}(\optoutcome) - \totalbid_{\leq \levelidx\setminus \bidderidx}^*(\optoutcome)) - (\totalbid_{\leq \levelidx\setminus \bidderidx}(\witness) - \totalbid_{\leq \levelidx\setminus \bidderidx}^*(\witness)) < \epsilon \highbound(\levelidx).
\end{equation}

Summing over equations \eqref{eq:smallerlevel-increasebound} and \eqref{eq:highbidderser} gives us:

\begin{align}
\label{eq:smallerbiglevel-increasebound}
&(\totalbid_{\leq \levelidx \setminus \bidderidx}(\optoutcome) - \totalbid_{\leq \levelidx\setminus \bidderidx}^*(\optoutcome)) - (\totalbid_{\leq \levelidx\setminus \bidderidx}(\witness) - \totalbid_{\leq \levelidx\setminus \bidderidx}^*(\witness)) \notag\\ 
&\ \ \ + (\totalbid_{>\levelidx}(\optoutcome) - \totalbid_{>\levelidx}^*(\optoutcome)) - (\totalbid_{>\levelidx}(\witness) - \totalbid_{>\levelidx}^*(\witness)) < \epsilon \highbound(\levelidx).
\end{align}

By assumption, $(b_{\bidderidx}(\optoutcome) - b_{\bidderidx}^*(\optoutcome)) = -(\utiltarget_{\bidderidx}-\utiltarget_{\bidderidx}^*) \leq - \epsilon \highbound(\levelidx)$ and $b_{\bidderidx}(\witness) = b_{\bidderidx}^*(\witness)) = 0$. Thus, $(b_{\bidderidx}(\optoutcome) - b_{\bidderidx}^*(\optoutcome)) - (b_{\bidderidx}(\witness) - b_{\bidderidx}i^*(\witness)) \leq - \epsilon \highbound(\levelidx)$. Adding this to \eqref{eq:smallerbiglevel-increasebound} gives:
\begin{equation}
(\totalbid(\optoutcome) - \totalbid^*(\optoutcome)) - (\totalbid(\witness) - \totalbid^*(\witness)) < \epsilon \highbound(\levelidx) - \epsilon \highbound(\levelidx)=0
\end{equation}
By our initial assumption that $\witness$ is a witness outcome, $\totalbid^*(\optoutcome) - \totalbid^*(\witness)=0$. Adding this to  the above equation yields
\begin{equation}
\totalbid(\optoutcome)  < \totalbid(\witness)
\end{equation}
This contradicts our assumption that $\optoutcome$ is a winning set with bids $b(\cdot)$. Hence, $\bidderidx$ will be forced to decrease her utility-target to at most $\utiltarget_{\bidderidx}^* + \epsilon \highbound(\levelidx)$ before the egalitarian winning set $\optoutcome$ is winning again. 
\end{proof}

Combining Lemma~\ref{lem:sbid-lowerbound} and Lemma~\ref{lem:sbid-upperbound} gives us convergence of each bidder in each level $\levelidx$ to within their bounds as soon as lower level bidders have all converged. It follows then from straightforward induction on levels that all bids converge to within their bounds.
\end{proofof}

\end{document}